\documentclass[11pt]{amsart}
\usepackage{graphicx,amssymb,amsmath,amsthm}
\usepackage{enumerate}
\usepackage{dsfont}
\usepackage[colorlinks, citecolor=red]{hyperref}
\usepackage{comment,cite,color}
\usepackage{cite,color}
\usepackage{mathrsfs}
\usepackage{epsfig}
\usepackage{lscape}
\usepackage{subfigure}
\usepackage{epstopdf}
\usepackage{caption}
\usepackage{bm}
\usepackage{algorithm}
\usepackage{algpseudocode}
\usepackage[english]{babel}
\usepackage{multirow}

\newcommand{\xkh}[1]{\left(#1\right)}
\newcommand{\dkh}[1]{\left\{#1\right\}}
\newcommand{\zkh}[1]{\left[#1\right]}

\newcommand{\nj}[1]{\langle {#1} \rangle}

\newcommand{\norm}[1]{\|{#1}\|_2}

\newcommand{\normone}[1]{\|{#1}\|_1}
\newcommand{\norms}[1]{\|{#1}\|}
\newcommand{\abs}[1]{\lvert#1\rvert}

\newcommand{\1}{{\mathds 1}}

\newcommand{\R}{{\mathbb R}}

\newcommand{\T}{\top}

\newcommand{\z}{{\widehat{\bm{z}}}}

\newcommand{\hu}{{\widehat{\bm{u}}}}
\newcommand{\hv}{{\widehat{\bm{v}}}}
\newcommand{\vx}{{\bm x}}
\newcommand{\vxs}{{{\bm x}}}
\newcommand{\vy}{{\bm y}}
\newcommand{\vu}{{\bm u}}
\newcommand{\vv}{{\bm v}}
\newcommand{\vw}{{\bm w}}
\newcommand{\vz}{{\bm z}}

\newcommand{\hS}{\widehat{ S }}

\newcommand{\vh}{{\bm h}}

\newcommand{\supp}{{\rm supp}}

\newcommand{\RNum}[1]{\uppercase\expandafter{\romannumeral #1\relax}}

\newtheorem{definition}{Definition}[section]

\newtheorem{prop}[definition]{Proposition}
\newtheorem{theorem}[definition]{Theorem}
\newtheorem{lemma}[definition]{Lemma}

\newtheorem{remark}[definition]{Remark}

\newtheorem{example}[definition]{Example}
\date{}

\begin{document}

\author{Meng Huang}
\address{School of Mathematical Sciences, Beihang University, Beijing, 100191, China} \email{menghuang@buaa.edu.cn}
\thanks{Meng Huang was supported by NSFC grant (12201022).}

\author{Shidong Li}
\address{Department of Mathematics, San Francisco State University, USA} \email{shidong@sfsu.edu}

\baselineskip 18pt
\bibliographystyle{plain}
\title[Recovery bounds for tail-minimization]{Performance analysis of tail-minimization and the linear rate of convergence of a proximal algorithm for sparse signal recovery}
\maketitle

\begin{abstract}
Recovery error bounds of tail-minimization and the rate of convergence of an efficient proximal alternating algorithm for sparse signal recovery are considered in this article. Tail-minimization focuses on minimizing the energy in the complement $T^c$ of an estimated support $T$. Under the restricted isometry property (RIP) condition, we prove that tail-$\ell_1$ minimization can exactly recover sparse signals in the noiseless case for a given $T$. In the noisy case, two recovery results for the tail-$\ell_1$ minimization and the tail-lasso models are established. Error bounds are improved over existing results. Additionally, we show that the RIP condition becomes surprisingly relaxed, allowing the RIP constant to approach $1$ as the estimation $T$ closely approximates the true support $S$.  Finally, an efficient proximal alternating minimization algorithm is introduced for solving the tail-lasso problem using Hadamard product parametrization. The linear rate of convergence  is established using the Kurdyka-{\L}ojasiewicz inequality.  Numerical results demonstrate that the proposed algorithm significantly improves signal recovery performance compared to state-of-the-art techniques.
\end{abstract}

\section{Introduction}
Let $\vx\in \R^n$ be a  sparse or nearly sparse vector. The problem of recovering $\vx$ from underdetermined system
\begin{equation} \label{eq:prob}
\vy=A\vxs+\vw,
\end{equation}
is termed as {\em compressed sensing} \cite{candes2006,donoho,foucar,rauhut}. Here, $A\in \R^{m\times n} (m<<n)$ is the measurement matrix, $\vy \in \R^m$ is the observed measurements, and $\vw \in \R^m$ is the noise vector. $\vx \in \R^n$ is said $s$-sparse if $\|\vx\|_0\leq s$ where $\|\vx\|_0$ denotes the number of nonzero entries of $\vx$.
Compressed sensing is ubiquitous in many areas of physical sciences and engineering, such as magnetic resonance imaging (MRI) \cite{Lustig,Lustig2}, computed tomography (CT) \cite{Sloun},  radar \cite{Ender,Herman}, statistics \cite{RIPcandes2,Bickel}, and others \cite{Duarte,Recht} to name just a few. We refer readers to a recent book \cite{foucar} for a comprehensive exposition of the subject.

The task of reconstructing $\vx$ can be formulated as the following $\ell_0$-minimization problem
\begin{equation} \label{eq:ell0}
  \min \limits_{\vz \in \R^n}  \norms{\vz}_0 \qquad \mbox{s.t.} \quad \norm{A\vz-\vy}\le \epsilon,
\end{equation}
where $\epsilon>0$ is the noise level. The $\ell_0$-minimization problem is known to be NP-hard \cite{natarajan} and un-tractable. A natural convex relaxation to (\ref{eq:ell0}) is the $\ell_1$-minimization problem,
\begin{equation}  \label{eq:ell1}
   \mathop{\min} \limits_{\vz \in \R^n}  \norms{\vz}_1 \qquad  \mbox{s.t.} \quad \norm{A\vz-\vy}\le \epsilon,
\end{equation}
which has effective solutions \cite{candes2006,donoho,chen2001}, etc.  Based on the restricted isometry property (RIP), or the null space property (NSP) of the matrix $A$, the equivalence between \eqref{eq:ell0} and \eqref{eq:ell1} has been established, e.g.,  \cite{RIPTcaiz,CJT,donoho2}.  Another approach is to solve the regularized  unconstrained program:
\begin{equation}
 \mathop{\min} \limits_{\vz \in \R^n}  \frac12 \norm{A\vz-\vy} +\lambda  \norms{\vz}_1,
\end{equation}
which is widely known as the LASSO problem (the least absolute shrinkage and selection operator) \cite{Tibshirani}, and can be resolved through iterative shrinkage algorithm \cite{Beck}.  Other algorithms include $\ell_q$-minimization ($0<q<1$), e.g., \cite{CY,XCX}, greedy  algorithms such as  OMP \cite{zhang2011sparse}, OMMP \cite{OMMP},  SP \cite{Dai}, CoSaMP \cite{needell2009cosamp}, iterative hard thresholding (IHT) \cite{blumensath2009iterative},  and hard thresholding pursuit (HTP) \cite{foucart2011hard}, and so on.

In this article, the following tail-$\ell_1$ minimization problem is considered to recover  the signal $\vx$:
\begin{equation} \label{eq:mod1}
\min_{\vz\in \R^n} \quad \normone{\vz_{T^c}}  \qquad \mbox{s.t.} \quad \norm{A\vz-\vy}\le \epsilon.
\end{equation}
Here, $T \subset [n]$ is an estimated index set of the true support $S$,  $T^c$ is the complement of $T$, and $\epsilon$ is the noise level.  The idea of tail minimization is to squeeze the energy on the complement $T^c$ to the estimated support set $T$, resulting in a sparser solution.   For the case where  the noise level $\epsilon$ is unknown,  we  consider the following  tail-lasso problem:
\begin{equation} \label{eq:mod2}
\min_{\vz \in \R^n} \qquad  \frac12\norm{A\vz-\vy}^2+ \lambda\normone{\vz_{T^c}},
\end{equation}
which is known to be equivalent to \eqref{eq:mod1} with  suitable parameter $\lambda>0$.

The effectiveness of tail-$\ell_1$ minimization was validated in \cite{lai2018spark}, where the authors demonstrated that it can recover signals with nearly spark-level sparsity. To provide theoretical guarantees, \cite{zheng2022mmv} established recovery error bounds for \eqref{eq:mod1} under the assumption of the tail Null Space Property (tail-NSP).  While the restricted isometry property (RIP) is another widely used assumption in compressed sensing analysis, the performance of tail-$\ell_1$ minimization under RIP condition remains unexplored.
In terms of algorithms,  since the programs \eqref{eq:mod1} and \eqref{eq:mod2} are convex, a wide range of algorithms are available for estimating their solutions ( see \cite{schmidt} for a review). Recently, the tail Hadamard Product Parametrization (tail-HPP) algorithm, initially proposed in \cite{hoff2017} for solving classical LASSO problems, has been investigated for tail-$\ell_1$ minimization in \cite{Guangxi}, which demonstrated superb speed of signal recovery among all solution techniques of the tail-$\ell_1$ minimization problem. However, theoretical guarantees for the convergence of the tail-HPP approach, especially concerning the rate of convergence for an estimated $T$, are still unclear.  In summary,
\begin{itemize}
\item[(1)]  Besides the (tail) NSP condition, the RIP condition also offers different and inside perspectives of sparse signal recovery. So far, the RIP analysis for tail minimization has not yet been explored. We aim at investigating the theoretical performance of  \eqref{eq:mod1} and \eqref{eq:mod2} under the widely used RIP condition in this article.  
\item[(2)]  Although the tail-HPP method has better recovery performance compared to state-of-the-art approaches, there is yet no convergence analysis. In this article, we aim to provide theoretical guarantees for an improved version of the tail-HPP algorithm and establish its linear rate of convergence for the problem \eqref{eq:mod2}.
\end{itemize}

Main contributions of this article are as follows:
\begin{itemize}
\item For a sparsity level $k$ and an estimated support $T\subset [n]$, we show that any solution $\z$ to \eqref{eq:mod1} with $\vy=A\vxs+\vw$ and $\norm{\vw} \le \epsilon$ obeys
\begin{equation*}
\norm{\z-\vxs} \le C_1 \epsilon+ C_2 \frac{\normone{\vxs_{T^c \cap S^c}} }{\sqrt{ak}},
\end{equation*}
provided the matrix $A$ satisfies a RIP condition. Here, $S$ is the index set of the $k$-largest  components of $\vxs$ in magnitude, and $C_1, C_2 >0$ are some constants.  We shall mention that the error bounds based on tail-NSP given in \cite{zheng2022mmv} are so far for individual recovery, while our results based on RIP condition holds for all signals $\vx\in \R^n$ of the order-$k$  sparsity.
\item When the matrix $A$ satisfies RIP condition with $0<\delta_{2k}<1$ and $T$ is a good estimate of the true support of  $\vx$, we show that the solution $\z$ to program \eqref{eq:mod1} is a reliable estimate of $\vx$.  Compared to existing results, such as $\delta_{2k}<\sqrt2-1$ in \cite{RIPcandes}, $\delta_{2k}<0.472$ in \cite{RIPTcai}, $\delta_{2k}<0.4931$ in \cite{Mo}, $\delta_{2k}<0.5$ in \cite{RIPTcai9}, and $\delta_{ak}<\sqrt{(a-1)/a}$ for $a\ge 4/3$ \cite{RIPTcaiz},   our result only requires  that any $2k$ columns of $A$ are linear independent, which is a significant improvement and aligns with the necessary condition for sparse recovery.
\item For any estimated support $T\subset [n]$ with $|T|=k_T <  m-k$,  any solution $\z$ to \eqref{eq:mod2} with parameter $\lambda > 2\norms{A^\T \vw}_\infty$  satisfies
\[
\norm{\z-\vx} \le C_1 \lambda \sqrt{k+k_T} + C_2 \sqrt{ \lambda \normone{\vxs_{T^c \cap S^c}} }+C_3 \frac{\normone{\vxs_{T^c \cap S^c}}}{\sqrt{ak}}
\]
for some positive constants $C_1, C_2$ and $C_3$, provided the matrix $A\in \R^{m\times n}$ satisfies a RIP condition. 
\item Finally, an improved version of the tail-HPP algorithm for solving the tail minimization problem \eqref{eq:mod2} is proposed.  A convergence analysis is provided, and its linear rate of convergence to the target solution is also established. Numerical experiments demonstrate that it achieves much higher successful rate compared to state-of-the-art techniques.
\end{itemize}

The rest of the paper is organized as follows: In Section 2, we present the performance error bounds for the tail-$\ell_1$ minimization \eqref{eq:mod1} and the tail-lasso problem \eqref{eq:mod2}. Comparisons with existing results are also discussed. In Section 3, we  propose a proximal alternating minimization algorithm for solving the tail minimization problem \eqref{eq:mod2} based on Hadamard product parametrization, and establish its linear rate of convergence using tools from the Kurdyka-{\L}ojasiewicz (KL) theory. Numerical experiments validating the effectiveness of the proposed  algorithm are provided in Section 4. A conclusive discussion is presented in Section 5.

The notations used in this paper are summarized as follows:  $T$ represents the estimated support, $S$ indexes the $k$-largest components of $\vxs$ in magnitude,  and $S_0$ indexes the $k$-largest components of $\vxs_{T^c}$ in magnitude.
 For an index set $T \subset [n]$,  $T^c$ denotes the complement of $T$,  and $A_T \in \R^{m\times n}$ denotes the submatrix consisting of the columns of matrix  $A \in \R^{m\times n}$ indexed by $T$.

\section{Performance Error Bounds}
In this section, we establish the error bounds for \eqref{eq:mod1} and \eqref{eq:mod2}. Compared to the tail-NSP condition for $T$ relative to $S$ used in \cite{zheng2022mmv}, we adopt the restricted isometry property (RIP) condition \cite{RIPcandes} for an order-$k$ sparsity recovery analysis.

\begin{definition}
The  matrix $A\in \R^{m\times n}$ satisfies order-$k$ Restricted Isometry Property (RIP) condition with constants $0<\delta_{k,l}<1$ and
$\delta_{k,u}>0$  if
\[
(1-\delta_{k,l}) \norm{\vz}^2 \le \norm{A\vz}^2 \le (1+\delta_{k,u}) \norm{\vz}^2
\]
holds simultaneous for all $k$-sparse signals $\vz\in \R^n$. We shall denote $\delta_k=\max\dkh{\delta_{k,l}, \delta_{k,u}}$.
\end{definition}

The next lemma is a direct consequence of the definition of RIP, which is given in \cite[Lemma 2.1]{RIPcandes}.

\begin{lemma}  \label{le:RIPcandes}
For any disjoint subset $S, T \subset \dkh{1,\ldots, n}$ with $|S| \le s$ and $|T|\le k$, if the matrix $A\in \R^{m\times n}$ obeys order-$(s+k)$ restricted isometry property with constants $0<\delta_{s+k,l}<1$ and  $\delta_{s+k,u}>0$, then
\[
\abs{\nj{A\vx,A\vx'}} \le \delta_{s+k} \norm{\vx}\norm{\vx'}
\]
holds for all $\vx, \vx' \in \R^n$ supported on $S$ and $T$, respectively. Here, $\delta_{s+k}=\max\dkh{\delta_{s+k,l}, \delta_{s+k,u}}$.
\end{lemma}

\subsection{Performance error bounds for constrained minimization}

We first provide an error bound for the tail-$\ell_1$ minimization \eqref{eq:mod1}, which implies that exact recovery can be achieved even when $T$ is chosen arbitrarily with $|T|< m-k$, given a mild RIP condition.
Throughout this paper, by order-$K$ RIP condition, we mean order-$\lceil  K \rceil $ when $K$ is a non-integer.

\begin{theorem} \label{th:1}
Suppose the estimated support $T \subset [n]$ with cardinality  $|T|=k_T$.  For any integer $k\ge 1$ and  any $a>0$, let $K=\max\dkh{(a+1)k+k_T, 2ak}$. Assume that  the matrix $A\in \R^{m\times n}$ satisfies order-$K$ RIP condition with constants $0<\delta_{K,l}<1$ and $\delta_{K,u}>0$. If
\begin{equation} \label{eq:raK}
a> \frac{2\delta_K^2}{(1-\delta_{K,l})^2},
\end{equation}
then the following holds: for all $\vxs \in \R^n$, the solution $\z$ to \eqref{eq:mod1} with $\vy=A\vxs+\vw$ and $\norm{\vw} \le \epsilon$ obeys
\begin{equation} \label{eq:Th12}
\norm{\z-\vxs} \le C_1 \epsilon+ C_2 \frac{\normone{\vxs_{T^c \cap S_0^c}} }{\sqrt{ak}}.
\end{equation}
Here, $S_0$ index the $k$-largest components of $\vxs_{T^c}$ in magnitude, $\delta_K=\max\dkh{\delta_{K,l},\delta_{K,u}}$  and
\[
C_1:= \frac{2(\sqrt{a}+1)\sqrt{1+\delta_K}}{\sqrt{a}(1-\delta_{K,l})-\sqrt2\delta_K}, \quad C_2:=\xkh{\frac{2\sqrt2 (\sqrt{a}+1)\delta_K}{\sqrt{a}(1-\delta_{K,l})-\sqrt2\delta_K} +2} .
\]
\end{theorem}

\begin{remark}
 Let $S$ be the index set of the $k$-largest components of $\vxs$ in magnitude. From the definition of $S_0$,  it is easy to see from Figure \ref{figure:1} that  $\normone{\vxs_{T^c \cap S_0^c}} \le \normone{\vxs_{T^c \cap S^c}}$.  Therefore, the solution $\z$ to \eqref{eq:mod1} also satisfies
 \begin{equation*}
\norm{\z-\vxs} \le C_1 \epsilon+ C_2 \frac{\normone{\vxs_{T^c \cap S^c}}}{\sqrt{ak}}.
\end{equation*}
 Furthermore, if $\vxs$ is a $k$-sparse signal, then it can be exactly recovered via tail-$\ell_1$ minimization \eqref{eq:mod1} in the absence of noise.
\end{remark}

\begin{figure}[H]
\centering
     \includegraphics[width=0.5\textwidth]{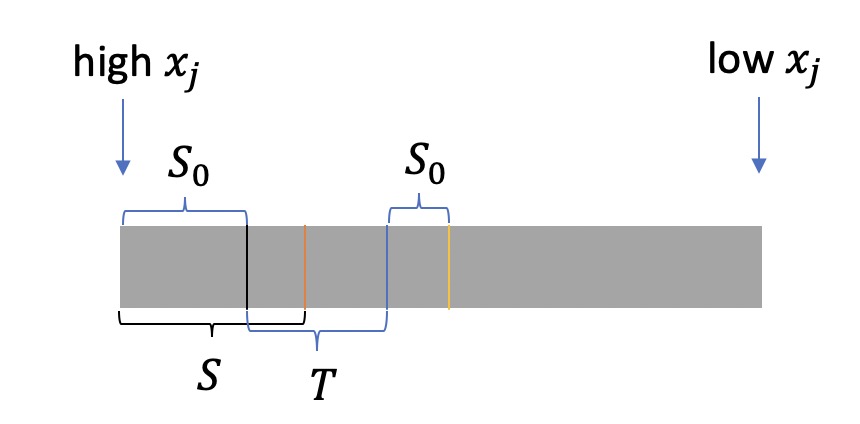}
\caption{The diagram of sets $S, T$, and $S_0$, where  $x_j$ are arranged in non-increasing order.}
\label{figure:1}
\end{figure}

\begin{remark}
The authors in \cite{zheng2022mmv} propose a  tail-NSP  condition and prove that if the matrix $A$ satisfies the tail-NSP for $T$ relative to $S \subset [n]$, then the solution $\z$ to \eqref{eq:mod1} obeys
\begin{equation} \label{eq:NSP}
\norm{\z-\vxs} \le C'_1 \epsilon+ C'_2 \norm{\vxs_{T^c \cap S^c}},
\end{equation}
where $C'_1$ and $ C'_2$ are positive constants  that depend on the matrix $A$.  In comparison, our result \eqref{eq:Th12} not only provides a slight improvement over \eqref{eq:NSP} but also applies universally to all $k$-sparse $\vx \in \R^n$, rather than for individual recovery as in \eqref{eq:NSP}.

\end{remark}

\begin{proof}[Proof of Theorem \ref{th:1}]
Set $\vh=\z-\vxs$.  For any $a>0$, we decompose $\vh$ into a sum of vectors $\vh_T, \vh_{S_0}, \vh_{S_1}, \vh_{S_2},\ldots$, where $S_0$ indexes  the $k$-largest  components of $\vxs_{T^c}$ in magnitude, $S_1$ corresponds to the locations of the $ak$-largest  of $\vh_{T^c \cap S_0^c}$ in magnitude, $S_2$ corresponds to the locations of the next $ak$-largest of $\vh_{T^c \cap S_0^c}$ in magnitude, and so on.  Observe that
\begin{equation} \label{eq:conp1}
\norm{\vh} \le  \norm{ \vh_{T\cup S_0\cup S_1}} +\sum_{j\ge 2} \norm{\vh_{S_j}}.
\end{equation}
We claim that the following holds:
\begin{equation} \label{eq:conp2}
\sum_{j\ge 2} \norm{\vh_{S_j}} \le \frac1{\sqrt a} \norm{ \vh_{S_0}} +  \frac{2\normone{\vxs_{T^c \cap S_0^c}} }{\sqrt{ak}}
\end{equation}
and
\begin{equation} \label{eq:conp3}
 \norm{ \vh_{T\cup S_0\cup S_1}} \le   \frac{2}{\sqrt{a}(1-\delta_{K,l})-\sqrt2\delta_K} \xkh{\sqrt{a(1+\delta_K)} \epsilon+ \frac{\sqrt2 \delta_K  \normone{\vxs_{T^c \cap S_0^c}} }{ \sqrt{k}}}.
\end{equation}
Putting \eqref{eq:conp2} and \eqref{eq:conp3} into \eqref{eq:conp1} and noting that $\norm{ \vh_{S_0}} \le  \norm{ \vh_{T\cup S_0\cup S_1}}$, we obtain that
\[
\norm{\vh} \le \frac{2(\sqrt{a}+1)\sqrt{1+\delta_K}}{\sqrt{a}(1-\delta_{K,l})-\sqrt2\delta_K}  \epsilon+ \xkh{\frac{2\sqrt2 (\sqrt{a}+1)\delta_K}{\sqrt{a}(1-\delta_{K,l})-\sqrt2\delta_K} +2}  \frac{\normone{\vxs_{T^c \cap S_0^c}} }{\sqrt{ak}}.
\]

It remains to prove the claims  \eqref{eq:conp2} and \eqref{eq:conp3}.  Since $\z$ is the solution to tail-$\ell_1$ minimization \eqref{eq:mod1} and $\vxs$ is a feasible point, we have
\begin{eqnarray*}
\normone{\vxs_{T^c}} \ge \normone{\z_{T^c}}& =&  \normone{\vh_{T^c}+\vxs_{T^c}} \\
&\ge & \normone{\vh_{T^c \cap S_0^c}+\vxs_{T^c \cap S_0} } - \normone{\vh_{T^c \cap S_0}+\vxs_{T^c \cap S_0^c} } \\
&=& \normone{ \vh_{T^c \cap S_0^c}}+ \normone{\vxs_{S_0}} - \normone{ \vh_{S_0}}- \normone{\vxs_{T^c \cap S_0^c}},
\end{eqnarray*}
where the last inequality comes from the fact that $S_0 \subset T^c$.
Combining with the decomposition that $\normone{\vxs_{T^c}}=\normone{\vxs_{S_0}} +  \normone{\vxs_{T^c \cap S_0^c}}$, one has
\begin{equation} \label{eq:hsc}
\normone{ \vh_{T^c \cap S_0^c}} \le \normone{ \vh_{S_0}}+ 2\normone{\vxs_{T^c \cap S_0^c}}.
\end{equation}
From the definition of $S_j$, it holds
\begin{equation} \label{eq:hsjsum}
\sum_{j\ge 2} \norm{\vh_{S_j}} \le \sum_{j\ge 2}  \sqrt{ak} \norms{\vh_{S_j}}_{\infty} \le \sum_{j\ge 2}  \frac{\normone{\vh_{S_{j-1}}}}{\sqrt{ak}} = \frac{\normone{\vh_{T^c \cap S_0^c}} }{\sqrt{ak}}.
\end{equation}
Putting \eqref{eq:hsc} into \eqref{eq:hsjsum}, claim \eqref{eq:conp2} follows, namely,
\begin{equation} \label{eq:cla11}
\sum_{j\ge 2} \norm{\vh_{S_j}} \le  \frac{\normone{ \vh_{S_0}}+ 2\normone{\vxs_{T^c \cap S_0^c}} }{\sqrt{ak}} \le \frac1{\sqrt a} \norm{ \vh_{S_0}} +  \frac{2\normone{\vxs_{T^c \cap S_0^c}} }{\sqrt{ak}},
\end{equation}
where the last inequality comes from the Cauchy-Schwarz inequality that $\normone{ \vh_{S_0}} \le \sqrt{k} \norm{ \vh_{S_0}}$.

We next turn to prove the claim \eqref{eq:conp3}.    By the triangle inequality,  one has
\begin{equation} \label{eq:Ah}
\norm{A\vh} =\norm{A(\z -\vxs)} \le \norm{A \z -\vy}+ \norm{A\vxs -\vy} \le 2\epsilon,
\end{equation}
where the last inequality follows from the fact that $\z$ is feasible for the program  \eqref{eq:mod1}. Note that $\vh_{T\cup S_0\cup S_1}=\vh- \sum_{j\ge 2} \vh_{S_j} $. From the restricted isometry property with $K=\max\dkh{(a+1)k+k_T, 2ak}$, we have
\begin{eqnarray}
(1-\delta_{K,l}) \norm{ \vh_{T\cup S_0\cup S_1}}^2  &\le&   \norm{A \vh_{T\cup S_0\cup S_1}}^2  \notag\\
&=& \nj{A \vh_{T\cup S_0\cup S_1}, A\vh}-\nj{A \vh_{T\cup S_0\cup S_1},  \sum_{j\ge 2}A \vh_{S_j} }  \notag\\
&\le & \sqrt{1+\delta_{K,u}} \norm{ \vh_{T\cup S_0\cup S_1}} \norm{A\vh} + \delta_{K} \xkh{\norm{\vh_{T\cup S_0}}+\norm{\vh_{S_1}}}  \sum_{j\ge 2} \norm{\vh_{S_j}} \notag \\
&\le & 2\epsilon \sqrt{1+\delta_{K}} \norm{ \vh_{T\cup S_0\cup S_1}} + \sqrt2 \delta_K  \norm{\vh_{T\cup S_0\cup S_1}} \sum_{j\ge 2} \norm{\vh_{S_j}}, \label{eq:AhTS12}
\end{eqnarray}
where we use Lemma \ref{le:RIPcandes} in the third line and  \eqref{eq:Ah} in the last line. Putting \eqref{eq:cla11} into \eqref{eq:AhTS12}, one immediately has
\[
 \norm{ \vh_{T\cup S_0\cup S_1}} \le \frac{2 \sqrt{1+\delta_K} }{1-\delta_{K,l}} \epsilon + \frac{\sqrt2 \delta_K}{\sqrt{a}(1-\delta_{K,l})} \norm{ \vh_{S_0}} +  \frac{2\sqrt2 \delta_K}{\sqrt{ak}(1-\delta_{K,l})} \normone{\vxs_{T^c \cap S_0^c}}.
\]
Note that $\norm{ \vh_{S_0}} \le  \norm{ \vh_{T\cup S_0\cup S_1}}$. Solving the above inequality gives
\[
 \norm{ \vh_{T\cup S_0\cup S_1}} \le   \frac{2}{\sqrt{a}(1-\delta_{K,l})-\sqrt2\delta_K} \xkh{\sqrt{a(1+\delta_K)} \epsilon+ \frac{\sqrt2 \delta_K  \normone{\vxs_{T^c \cap S_0^c}} }{ \sqrt{k}}}.
\]
This completes the claim \eqref{eq:conp3}.
\end{proof}

The next theorem demonstrates that if $T$ is a good estimate of the true support $S$, then the solution $\z$ to tail-$\ell_1$ minimization \eqref{eq:mod1} serves as a reliable estimate of the true signal $\vx$ even when the RIP constant approaches $1$, which is a significant improvement over RIP constants for traditional $\ell_1$ minimization.

\begin{theorem} \label{th:0}
Suppose $T \subset [n]$ with $|T|=k_T$. For any $a>0$,  let $K'=\max\dkh{ak+k_T, 2ak}$.  Assume that  the matrix $A\in \R^{m\times n}$ satisfies order-$K'$ RIP condition with constants $0<\delta_{K',l}<1$ and $\delta_{K',u}>0$. Then the following holds: for all $\vxs \in \R^n$, the solution $\z$ to \eqref{eq:mod1} with $\vy=A\vxs+\vw$ and $\norm{\vw} \le \epsilon$ obeys
\[
\norm{\z-\vxs} \le C_1 \epsilon+ C_2 \frac{\normone{\vxs_{T^c}} }{\sqrt{ak}}.
\]
Here,
\[
C_1:=  \frac{2 \sqrt{1+\delta_{K'}} }{1-\delta_{K', l}}, \quad C_2:=\frac{2\sqrt2 \delta_{K'}}{1-\delta_{K',l}} +2
\]
with $\delta_{K'}=\max\dkh{\delta_{K', l},\delta_{K', u}}$.
\end{theorem}

\begin{remark}
Setting $k_T=k$ and considering the case where $0<a\le 1$ in Theorem \ref{th:0}, when $T$ is a good estimate of the support $S$,  the target signal $\vxs$ can be stably reconstructed by tail-$\ell_1$ minimization, provided that the matrix $A$ satisfies order-$(a+1)k$ RIP condition with constant $0<\delta_{(a+1)k}<1$.   

Compared to existing results, such as $\delta_{2k}<\sqrt2-1$ in {\rm \cite{RIPcandes}}, $\delta_{2k}<0.472$ in {\rm \cite{RIPTcai}}, $\delta_{2k}<0.4931$ in {\rm \cite{Mo}}, $\delta_{2k}<0.5$ in {\rm \cite{RIPTcai9}}, and $\delta_{ak}<\sqrt{(a-1)/a}$ for $a\ge 4/3$ {\rm \cite{RIPTcaiz}},   our result only requires  that any $\lceil (a+1)k \rceil$ columns of $A$ are linear independent for $a>0$, which reaches nearly spark-level sparsity recovery performance. 

To validate our theoretical results, in Figure  \ref{figure:rip}, we plot the RIP constant curve for a Gaussian matrix $A\in \R^{64\times 256}$ with each entry  $a_{i,j} \sim \mathcal N(0,1/m)$  and record the success rates of  tail-$\ell_1$ minimization \eqref{eq:mod1}, comparing it to classical $\ell_1$ minimization \eqref{eq:ell1} across different sparsity levels. The success rates are computed from $1000$ independent trials, with a trial considered successful if the solution $\z$  satisfying  $\norm{\z-\vx}\le 10^{-5}$. Both \eqref{eq:ell1} and \eqref{eq:mod1} are solved using the CVX toolbox. Figure \ref{figure:ripa} shows that when the sparsity
 $k$ is around $50$, the RIP constant is close to $1$. However, as presented in Figure \ref{figure:ripb},  tail-$\ell_1$ minimization \eqref{eq:mod1} still performs well, provided $T$ is a good approximation of the true support $S$.
\end{remark}

\begin{figure}[H]
\centering
\subfigure[]{
     \includegraphics[width=0.4\textwidth]{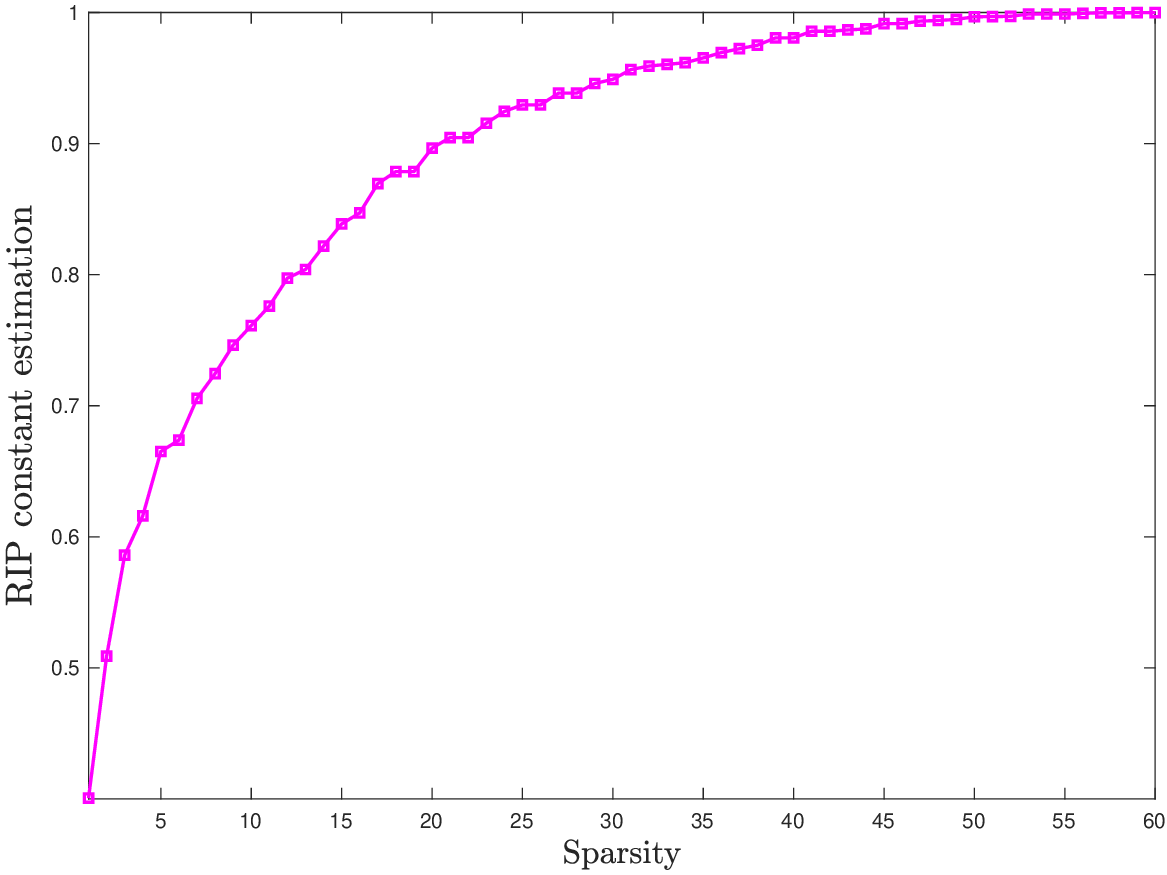}
     \label{figure:ripa}}
\subfigure[]{
     \includegraphics[width=0.42\textwidth]{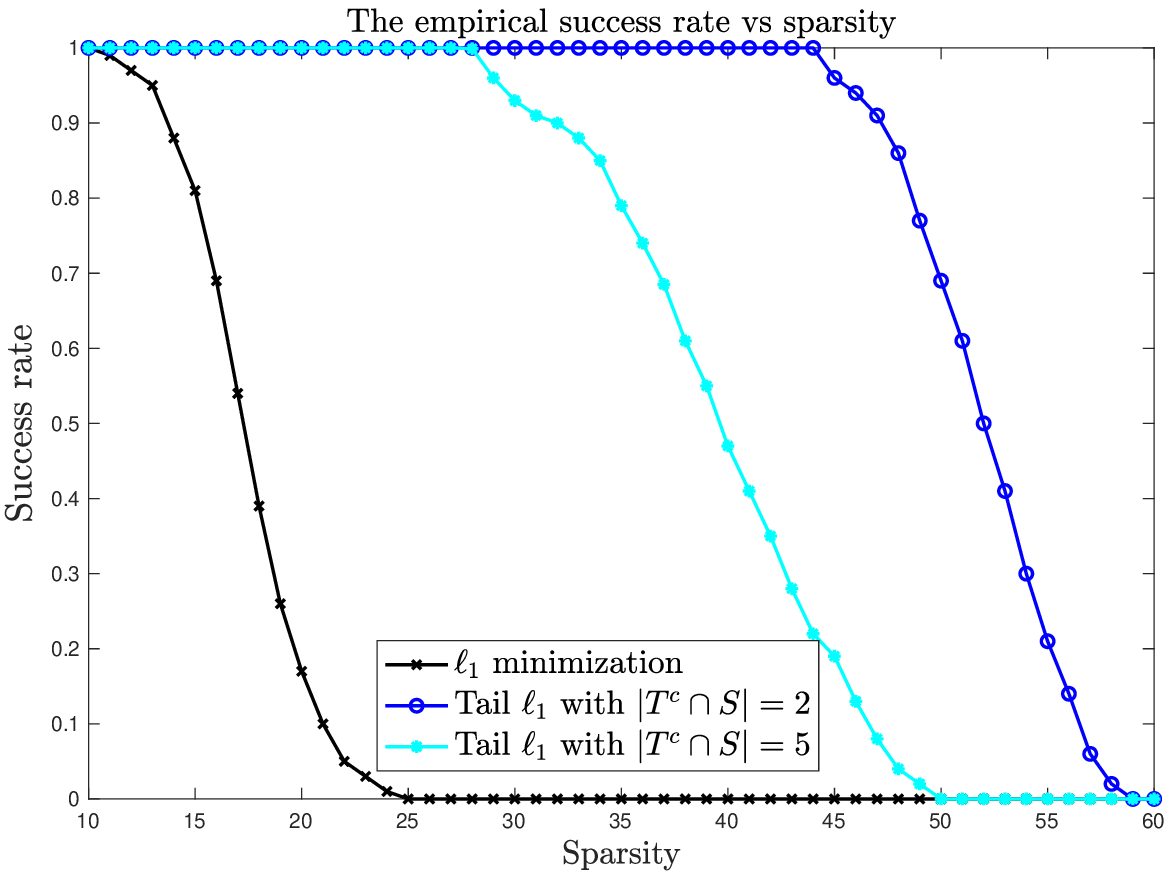}
     \label{figure:ripb}}
\caption{Recovery performance of tail-$\ell_1$ minimization and the classical $\ell_1$ minimization under different RIP constants recorded in the same set experiments: (a) The curve of RIP constant $\delta_{k,l}$ vs sparsity $k$; (b) The success rate vs sparsity in noiseless case. Here, the measurement matrix $A\in \R^{64\times 256}$.}
\label{figure:rip}
\end{figure}

\begin{remark} \label{re:28}
To guarantee the effectiveness of Theorem \ref{th:0}, the estimated support $T$ must closely approximate the true support $S$ of  $\vx$. Our numerical experiments show that if $T$ is selected using an adaptive strategy where $T_1$ indexes the $k_0$ (normally $k_0=1$) largest components in magnitude of the first iteration $\vz_1$, $T_2$ indexes the $k_0+k'$ (typically $k'=1$) largest  components in magnitude of the second iteration $\vz_2$,  $T_3$ indexes the $k_0+2k'$ largest components of the third iteration $\vz_3$ in magnitude, and so on, then after a few iterations, the set $T_k$ is seen to converge towards the true support $S$ of  $\vx$.
\end{remark}

\begin{proof}[Proof of Theorem \ref{th:0}]
Set $\vh=\z-\vxs$. Since $\z$ is the solution to the tail-$\ell_1$ minimization \eqref{eq:mod1}, we have
\begin{eqnarray*}
\normone{\vxs_{T^c}} \ge \normone{\z_{T^c}} =  \normone{\vh_{T^c}+\vxs_{T^c}} \ge   \normone{\vh_{T^c}} - \normone{\vxs_{T^c}}
\end{eqnarray*}
It immediately gives
\begin{equation} \label{eq:hsc12}
\normone{ \vh_{T^c}} \le 2\normone{\vxs_{T^c}}.
\end{equation}
Next, for any $a>0$, we decompose $\vh$ into a sum of vectors $\vh_{T}, \vh_{S_1}, \vh_{S_2},\ldots$, where $S_1$ corresponds to the locations of the $ak$ largest-magnitude of $\vh_{T^c}$, $S_2$ corresponds to the locations of the next $ak$ largest-magnitude of $\vh_{T^c}$, and so on.  Observe that
\begin{equation} \label{eq:conp112}
\norm{\vh} \le  \norm{ \vh_{T \cup S_1}} +\sum_{j\ge 2} \norm{\vh_{S_j}}.
\end{equation}
For the second term in right hand side of \eqref{eq:conp112}, it follows from the definition of $S_j$  that
\begin{equation} \label{eq:hsjsum12}
\sum_{j\ge 2} \norm{\vh_{S_j}} \le \sum_{j\ge 2}  \sqrt{ak} \norms{\vh_{S_j}}_{\infty} \le \sum_{j\ge 2}  \frac{\normone{\vh_{S_{j-1}}}}{\sqrt{ak}} = \frac{\normone{\vh_{T^c}} }{\sqrt{ak}} \le  \frac{ 2\normone{\vxs_{T^c}} }{\sqrt{ak}},
\end{equation}
where the last inequality comes from \eqref{eq:hsc12}.  To give an upper bound for the first term in right hand side of \eqref{eq:conp112}, due to the restricted isometry property of the matrix $A$, we have
\begin{equation} \label{eq:AhTS12121}
  \norm{A \vh_{T\cup S_1}}^2 \ge (1-\delta_{K',l}) \norm{ \vh_{T\cup S_1}}^2.
\end{equation}
Here,  $K'=\max\dkh{ak+k_T, 2ak}$. Therefore, it suffices to find an upper bound of $\norm{A \vh_{T\cup S_1}}^2$. Note that $\vh_{T \cup S_1}=\vh- \sum_{j\ge 2} \vh_{S_j} $.  One has
\begin{eqnarray}
 \norm{A \vh_{T\cup S_1}}^2 &=& \nj{A \vh_{T\cup S_1}, A\vh}-\nj{A \vh_{T\cup S_1},  \sum_{j\ge 2} A \vh_{S_j} }  \notag\\
&\le & \sqrt{1+\delta_{K'}} \norm{A \vh_{T \cup S_1}} \norm{A\vh} + \delta_{K'} \xkh{\norm{\vh_{T}}+\norm{\vh_{S_1}}}  \sum_{j\ge 2} \norm{\vh_{S_j}} \notag \\
&\le & 2\epsilon \sqrt{1+\delta_{K'}} \norm{ \vh_{T \cup S_1}} + \sqrt2 \delta_{K'}  \norm{ \vh_{T \cup S_1}} \sum_{j\ge 2} \norm{\vh_{S_j}}, \label{eq:AhTS1212}
\end{eqnarray}
where the first inequality follows from  Lemma \ref{le:RIPcandes} and the fact that
\begin{equation*}
\norm{A\vh} =\norm{A(\z -\vxs)} \le \norm{A \z -\vy}+ \norm{A\vxs -\vy} \le 2\epsilon.
\end{equation*}
Here, the first inequality arises from the triangle inequality and the last inequality follows from the fact that $\z$ is feasible for the problem  \eqref{eq:mod1}. Combining \eqref{eq:hsjsum12}, \eqref{eq:AhTS12121} and \eqref{eq:AhTS1212} together, we immediately obtain
\begin{equation} \label{eq:conp312}
  \norm{ \vh_{T\cup S_1}} \le \frac{2 \sqrt{1+\delta_{K'}} }{1-\delta_{K',l}} \epsilon +  \frac{2\sqrt2 \delta_{K'}}{\sqrt{ak}(1-\delta_{K',l})} \normone{\vxs_{T^c}}.
\end{equation}
Putting \eqref{eq:hsjsum12} and \eqref{eq:conp312} into \eqref{eq:conp112}, we obtain that
\[
\norm{\vh} \le  \frac{2 \sqrt{1+\delta_{K'}} }{1-\delta_{K',l}} \epsilon +  \xkh{ \frac{2\sqrt2 \delta_{K'}}{1-\delta_{K',l}} +2 } \frac{\normone{\vxs_{T^c}} }{\sqrt{ak}}.
\]
This completes the proof.

\end{proof}

%

In Theorem \ref{th:1}, when $T=\emptyset$ and $a=1$,  the next example shows that  the traditional result in \cite{RIPcandes} is re-produced for recovering the sparse signals through $\ell_1$ minimization program.

\begin{example} \label{exap:1}
Assume that  the matrix $A\in \R^{m\times n}$ satisfies order-$2k$ RIP condition with constant $\delta_{2k}<\sqrt2-1$. Then  the following holds: for all $\vxs \in \R^n$, the solution $\z$ to $\ell_1$-minimization
\[
\min_{\vz \in \R^n } \qquad \normone{\vz }  \quad \mbox{s.t.} \quad \norm{A\vz-\vy}\le \epsilon
\]
 with $\vy=A\vxs+\vw$ and $\norm{\vw} \le \epsilon$ obeys
\[
\norm{\z-\vxs} \le C_1 \epsilon+ C_2 \frac{\normone{\vxs_{S^c}} }{\sqrt{k}}.
\]
Here,  $S$ index the $k$-largest-magnitude components of $\vxs$, and
\[
C_1:= \frac{ 4 \sqrt{1+\delta_{2k}}}{1-(\sqrt2+1)\delta_{2k}}, \quad C_2:=\xkh{\frac{4\sqrt2 \delta_{2k}}{1-(\sqrt2+1)\delta_{2k}} +2} .
\]
\end{example}

\begin{remark}
Note that the optimal RIP constant for $\ell_1$-minimization is $\delta_{2k} < 1/\sqrt2$ {\rm \cite{RIPTcaiz}}.  Example \ref{exap:1} implies  that the RIP constant $\delta_K$ given in \eqref{eq:raK} may not be optimal. This can be improved using the method outlined in {\rm \cite{RIPTcaiz}}, which is left in subsequent future work.
\end{remark}

\subsection{The error bounds for tail-lasso}


In this subsection, we give the performance analysis for the unconstrained tail-$\ell_1$ minimization \eqref{eq:mod2}.

\begin{theorem} \label{th:2}
Suppose the estimated support $T \subset [n]$ with $|T|=k_T$. For any $a>0$,  let $K=\max\dkh{(a+1)k+k_T, 2ak}$ and $\gamma=\sqrt{\frac{1+k_T/k}{a}}$. Assume that the matrix $A\in \R^{m\times n}$ satisfies order-$K$ RIP condition with constant $\delta_K$ obeying
\[
\delta_K+ 3\sqrt2  \gamma \delta_K <1.
\]
Then the following holds: for all $\vxs \in \R^n$, the solution $\z$ to \eqref{eq:mod2} with $\lambda > 2\norms{A^\T \vw}_\infty$ and $\vy=A\vxs+\vw$ satisfies
\begin{equation} \label{eq:zhonn}
\norm{\z-\vx} \le C_1 \lambda \sqrt{k+k_T} + C_2 \sqrt{ \lambda \normone{\vxs_{T^c \cap S_0^c}} }+C_3 \frac{\normone{\vxs_{T^c \cap S_0^c}}}{\sqrt{ak}}.
\end{equation}
Here, $S_0$ index the $k$-largest  components of $\vxs_{T^c}$ in magnitude, and
\[
  C_1:= \frac{6(1+\delta_K)\xkh{3  \alpha +1} }{\xkh{ 1-\beta }^2}, \qquad C_2:=  \frac{2\sqrt{2}\xkh{3  \alpha +1} }{ 1-\beta }, \qquad C_3:=  \frac{8\delta_K \xkh{3  \alpha +1} }{ 1-\beta }+4,
\]
where $\beta:=\delta_K+ 3\sqrt2  \gamma \delta_K<1$.
\end{theorem}

\begin{remark}
The error bound \eqref{eq:zhonn} is optimal for $k$-sparse signals. To see this, consider  the Gaussian noise vector $\vw \sim \frac{\sigma}{\sqrt{m}}\cdot \mathcal{N}(0, I_m)$,  if the columns of the matrix $A$ are normalized, then with high probability $\norms{A^\T \vw}_\infty =O(\sigma\sqrt{\frac{\log n}{m}})$.   For this case,  if we take $\lambda=O(\sigma\sqrt{\frac{\log n}{m}})$, then the error bound \eqref{eq:zhonn} is $\norm{\z-\vxs} \le O(\sigma \sqrt{\frac{k\log n}{m}})$ for all $k$-sparse signals $\vxs$, which is optimal.
\end{remark}

\begin{proof}[Proof of Theorem \ref{th:2}]
Set $\vh=\z-\vxs$.  For any $a>0$,  let $\vh=\vh_T+\vh_{S_0} + \sum_{j\ge 1} \vh_{S_j}$, where $S_0$ indexes  the $k$-largest components of $\vxs_{T^c}$ in magnitude, $S_1$ corresponds to the locations of the $ak$-largest of $\vh_{T^c \cap S_0^c}$in magnitude, $S_2$ corresponds to the locations of the next $ak$-largest of $\vh_{T^c \cap S_0^c}$ in magnitude, and so on.  By the triangle inequality, it is easy to see that
\begin{equation} \label{eq:conp2117}
\norm{\vh} \le  \norm{ \vh_{T\cup S_0\cup S_1}} +\sum_{j\ge 2} \norm{\vh_{S_j}}.
\end{equation}
We claim that the following holds:
\begin{equation} \label{eq:conp2217}
\sum_{j\ge 2} \norm{\vh_{S_j}} \le 3 \sqrt{\frac{1+k_T/k}{a}}\norm{ \vh_{T\cup S_0}} +  \frac{4\normone{\vxs_{T^c \cap S_0^c}} }{\sqrt{ak}}
\end{equation}
and
\begin{equation} \label{eq:conp2317}
 \norm{ \vh_{T\cup S_0\cup S_1}} \le   \frac{6(1+\delta_K) \lambda \sqrt{k+k_T}}{(1-\beta)^2} + \frac{2\sqrt{2\lambda }}{1-\beta} \sqrt{\normone{\vxs_{T^c \cap S_0^c}} }+\frac{8\delta_K}{(1-\beta) \sqrt{ak}} \normone{\vxs_{T^c \cap S_0^c}}.
\end{equation}
Here, $\beta:=\delta_K+ 3\sqrt2 \delta_K \sqrt{\frac{1+k_T/k}{a}}<1$.
Putting \eqref{eq:conp2217} and \eqref{eq:conp2317} into \eqref{eq:conp2117}, we obtain the conclusion that
\[
\norm{\vh} \le C_1 \lambda \sqrt{k+k_T} + C_2 \sqrt{ \lambda \normone{\vxs_{T^c \cap S_0^c}} }+C_3 \frac{\normone{\vxs_{T^c \cap S_0^c}}}{\sqrt{ak}},
\]
where
\[
  C_1:= \frac{6(1+\delta_K)\xkh{3  \gamma +1} }{\xkh{ 1-\beta }^2}, \qquad C_2:=  \frac{2\sqrt{2}\xkh{3  \gamma +1} }{ 1-\beta }, \qquad C_3:=  \frac{8\delta_K \xkh{3  \gamma +1} }{ 1-\beta }+4
\]
with $\gamma=\sqrt{\frac{1+k_T/k}{a}}$.

For the claim \eqref{eq:conp2217}, since $\z$ is the solution to tail-lasso \eqref{eq:mod2}, one has
\[
\frac12\norm{A\z-\vy}^2+ \lambda\normone{\z_{T^c}} \le \frac12\norm{A\vxs-\vy}^2+ \lambda\normone{\vxs_{T^c}}.
\]
Note that $\vy=A\vxs+\vw$ and $\vh=\z-\vxs$. A simple calculation gives
\begin{eqnarray*}
\norm{A\vh}^2 & \le &  2\nj{\vh, A^\T\vw} +2\lambda\xkh{\normone{\vxs_{T^c}}- \normone{\z_{T^c}}} \\
&\le & 2 \norms{A^\T \vw}_\infty \normone{\vh} +2\lambda\xkh{\normone{\vxs_{T^c}}- \normone{\z_{T^c}}}.
\end{eqnarray*}
Since
\begin{eqnarray*}
\normone{\z_{T^c}} = \normone{\vh_{T^c}+\vxs_{T^c}} &\ge & \normone{\vh_{T^c \cap S_0^c}+\vxs_{T^c \cap S_0} } - \normone{\vh_{T^c \cap S_0}+\vxs_{T^c \cap S_0^c} } \\
&=& \normone{ \vh_{T^c \cap S_0^c}}+ \normone{\vxs_{ S_0}} - \normone{ \vh_{ S_0}}- \normone{\vxs_{T^c \cap S_0^c}}
\end{eqnarray*}
and $\normone{\vxs_{T^c}} = \normone{\vxs_{S_0}} + \normone{\vxs_{T^c \cap S_0^c}}$, it then follows from the fact $2\norms{A^\T \vw}_\infty < \lambda $ that
\begin{eqnarray}
0\le \norm{A\vh}^2 &\le  & \lambda \normone{\vh}+2\lambda\xkh{\normone{ \vh_{S_0}}-  \normone{ \vh_{T^c \cap S_0^c}}+ 2\normone{\vxs_{T^c \cap S_0^c}}} \notag \\
&\le  &  \lambda\xkh{3\normone{\vh_{T \cup S_0}}-  \normone{ \vh_{T^c \cap S_0^c}}+ 4\normone{\vxs_{T^c \cap S_0^c}} }. \label{eq:Ahlam}
\end{eqnarray}
This immediately gives
\begin{equation} \label{eq:hsc111}
\normone{ \vh_{T^c \cap S_0^c}} \le 3\normone{ \vh_{T \cup S_0}}+ 4\normone{\vxs_{T^c \cap S_0^c}}.
\end{equation}
From the definition of $S_j$ , we know
\begin{equation} \label{eq:hsjsum111}
\sum_{j\ge 2} \norm{\vh_{S_j}} \le \sum_{j\ge 2}  \sqrt{ak} \norms{\vh_{S_j}}_{\infty} \le \sum_{j\ge 2}  \frac{\normone{\vh_{S_{j-1}}}}{\sqrt{ak}} = \frac{\normone{\vh_{T^c \cap S_0^c}} }{\sqrt{ak}}.
\end{equation}
Putting \eqref{eq:hsc111} into \eqref{eq:hsjsum111}, we obtain the conclusion of claim \eqref{eq:conp2217}, namely,
\begin{equation} \label{eq:cla1}
\sum_{j\ge 2} \norm{\vh_{S_j}} \le  \frac{3\normone{ \vh_{T \cup S_0}}+ 4\normone{\vxs_{T^c \cap S_0^c}} }{\sqrt{ak}} \le 3\sqrt{\frac{1+k_T/k}{a}}\norm{ \vh_{T\cup S_0}} +  \frac{4\normone{\vxs_{T^c \cap S_0^c}} }{\sqrt{ak}},
\end{equation}
where the last inequality comes from the Cauchy-Schwarz inequality and the fact that $|T\cup S_0|=k+k_T$.

For the claim \eqref{eq:conp2317},   it follows from \eqref{eq:Ahlam} that
\begin{equation} \label{eq:lamres}
\norm{A\vh}^2 \le 3\lambda \sqrt{k+k_T} \norm{\vh_{T \cup S_0}}+ 4\lambda \normone{\vxs_{T^c \cap S_0^c}}.
\end{equation}
Note that $\vh_{T\cup S_0\cup S_1}=\vh- \sum_{j\ge 2} \vh_{S_j} $. From the restricted isometry property with $K=\max\dkh{(a+1)k+k_T, 2ak}$, we have
\begin{eqnarray}
(1-\delta_K) \norm{ \vh_{T\cup S_0\cup S_1}}^2  &\le&   \norm{A \vh_{T\cup S_0\cup S_1}}^2  \notag\\
&=& \nj{A \vh_{T\cup S_0\cup S_1}, A\vh}-\nj{A \vh_{T\cup S_0\cup S_1},  \sum_{j\ge 2} A \vh_{S_j} }  \notag\\
&\le & \sqrt{1+\delta_K} \norm{ \vh_{T\cup S_0\cup S_1}} \norm{A\vh} +\sqrt 2 \delta_{K} \norm{ \vh_{T\cup S_0\cup S_1}}    \sum_{j\ge 2} \norm{\vh_{S_j}} \label{eq:AhTS1225}.
\end{eqnarray}
Putting \eqref{eq:hsjsum111} and \eqref{eq:lamres} into \eqref{eq:AhTS1225} and noting that $\norm{ \vh_{T\cup S_0}}  \le \norm{ \vh_{T\cup S_0\cup S_1}}$, we have
\begin{eqnarray*}
&& \xkh{1-\delta_K-3\sqrt2 \delta_K \sqrt{\frac{1+k_T/k}{a}} } \norm{ \vh_{T\cup S_0\cup S_1}}  \\
&\le &  \sqrt{1+\delta_K}  \sqrt{ 3\lambda \sqrt{k+k_T} \norm{\vh_{T \cup S_0}}+ 4\lambda \normone{\vxs_{T^c \cap S_0^c}} }+  \frac{4\sqrt2 \delta_K\normone{\vxs_{T^c \cap S_0^c}} }{\sqrt{ak}}
\end{eqnarray*}
Recall that  $\beta=1-\delta_K-3\sqrt2 \delta_K \sqrt{\frac{1+k_T/k}{a}}>0$.   Solving the above inequality gives
\[
 \norm{ \vh_{T\cup S_0\cup S_1}} \le   \frac{6(1+\delta_K) \lambda \sqrt{k+k_T}}{\beta^2} + \frac{2\sqrt{2\lambda }}{\beta} \sqrt{\normone{\vxs_{T^c \cap S_0^c}} }+\frac{8\delta_K}{\beta \sqrt{ak}} \normone{\vxs_{T^c \cap S_0^c}}.
 \]
This completes the proof.

\end{proof}

\section{Proximal HPP algorithm for solving tail-lasso problems}
In this section, we introduce a fast proximal alternating minimization algorithm for solving the tail-lasso problems \eqref{eq:mod2} based on Hadamard product parametrization (HPP).  The convergence analysis is provided for \eqref{eq:mod2}. The linear rate of convergence of the algorithm is also established.

\subsection{The tail-HPP algorithm for tail-lasso}

Recall the tail-lasso problem is to find the solution of:
\begin{equation} \label{eq:orhpp}
 \min_{\vz\in \R^n} \qquad  f(\vz):=\frac12\norm{A\vz-\vy}^2+ \lambda\normone{\vz_{T^c}},
\end{equation}
where $\lambda>0$ is a parameter which depends on the desired level of sparsity. 
Its associated over-parameterized model is
\begin{equation} \label{eq:hpp}
 \min_{\vu, \vv \in \R^n} \qquad  g(\vu,\vv):=\frac12  \norm{A(\vu \circ \vv) -\vy  }^2+  \frac{\lambda}2 \xkh{\norm{\vu_{T^c}}^2+\norm{\vv_{T^c}}^2},
\end{equation}
where $ \vu \circ \vv$ denotes the Hadamard (element-wise) product between $\vu$ and $\vv$.  Problem \eqref{eq:hpp} was referred to as the  tail-Hadamard product parametrization of \eqref{eq:orhpp} in \cite{Guangxi}, and the idea was deeply inspired by the recent works \cite{hoff2017,poon2023}.  Although the loss function $g$ is non-convex, it is however differentiable, and its local minimizers can be found using a fast proximal alternating minimization algorithm. Furthermore, as we will show later, for any local minimizer $(\hu,\hv)$ of $g(\vu,\vv)$,  the point $\z:=\hu\circ\hv$ is the global minimizer of $f(\vz)$. Based on this, it suffices to find a local minimizer of $g(\vu,\vv)$.  The following lemma, originally presented in \cite{Guangxi} with a proof adapted from \cite{hoff2017}, is provided here with a more intuitive proof.

\begin{lemma} \label{le:sameinf}
For any matrix $A\in \R^{m\times n}$ and any vector $\vy \in \R^m$, let $f(\vz)$ and $g(\vu,\vv)$ be defined as \eqref{eq:orhpp} and \eqref{eq:hpp}, respectively. Then $\inf_{\vz}~f(\vz)=\inf_{\vu,\vv}~g(\vu,\vv)$.
\end{lemma}
\begin{proof}
We first prove that $\inf_{\vz}~f(\vz) \le \inf_{\vu,\vv}~g(\vu,\vv)$.   Note that
\[
g(\vu,\vv)=\frac12  \norm{A(\vu \circ \vv) -\vy  }^2+  \frac{\lambda}2 \xkh{\norm{\vu_{T^c}}^2+\norm{\vv_{T^c}}^2}
\]
 is a continuous function. We claim that there exists $(\hu,\hv) \in \R^n \times \R^n$ such that $\inf_{\vu,\vv}~g(\vu,\vv)=g(\hu,\hv)$.  Indeed,  for any $(\vu,\vv) \in \R^n \times \R^n $, if $A(\vu \circ \vv) \neq \bm 0$ then  $g(t \vu,t \vv) \to +\infty$ as $t \to \infty$. And if $A(\vu \circ \vv) = \bm 0$ then $g(t \vu,t \vv) \ge g(\vu, \vv) $ for all $t\ge 1$. Therefore, the minimum of $g(\vu, \vv)$ can be reached at some point $(\hu,\hv) \in \R^n \times \R^n$. This gives the claim.
Set $\z=\hu \circ \hv$.  Then we have
\begin{eqnarray*}
f(\z) &= & \frac12 \norm{ A (\hu \circ \hv) -\vy}^2 +  \lambda \normone{(\hu \circ \hv)_{T^c}} \\
&\le & \frac12 \norm{ A (\hu \circ \hv) -\vy}^2  +  \lambda \cdot \frac{\norm{\hu_{T^c}}^2+\norm{\hv_{T^c}}^2}{2} \\
&=& g(\hu,\hv),
\end{eqnarray*}
where the inequality comes from the fact that $ab\le \frac{a^2+b^2}{2}$ for any $a,b\in \R$. This implies $\inf_{\vz}~f(\vz) \le f(\z)  \le g(\hu,\hv) = \inf_{\vu,\vv}~g(\vu,\vv)$.

We next turn to show that $\inf_{\vu,\vv}~g(\vu,\vv) \le \inf_{\vz}~f(\vz)$. Similarly, we assume that $\inf_{\vz}~f(\vz)=f(\z)$ for some $\z\in \R^n$. Define the vectors $\hu,\hv\in \R^n$ as follows:
\[
\widehat{u}_j=\mbox{sign}(\widehat{z}_j)  \sqrt{|\widehat{z}_j|} \quad  \quad  \widehat{v}_j= \sqrt{|\widehat{z}_j|}  \qquad \mbox{for all} \qquad j=1,\ldots, n.
\]
Observe that $\hu\circ \hv=\z$ and $\abs{\widehat{u}_j}^2+\abs{ \widehat{v}_j }^2 =2\abs{\widehat{z}_j}$. It then gives
\begin{eqnarray*}
g(\hu,\hv) &=&  \frac12  \norm{ A \z -\vy}^2 + \frac{\lambda}2 \xkh{\norm{\hu_{T^c}}^2+\norm{\hv_{T^c}}^2} =f(\z).
\end{eqnarray*}
Therefore, we have $\inf_{\vu,\vv}~g(\vu,\vv) \le  g(\hu,\hv)=f(\z) = \inf_{\vz}~f(\vz)$. This completes the proof.
\end{proof}

Next, we establish a correspondence between the minimizers of $f$ and $g$. Additionally, we prove that  the loss function $g$  exhibits a benign geometric landscape, meaning that all local minimizers are global minimizers and each saddle point has negative curvature direction. These favorable properties will facilitate finding globally optimal solution of $g$.

\begin{theorem} \label{th:geo}
Let $f(\vz)$ and $g(\vu,\vv)$ be defined as \eqref{eq:orhpp} and \eqref{eq:hpp}, respectively. If $(\hu,\hv)$ is a local minimizer of $g(\vu,\vv)$, then $\z:=\hu\circ\hv$ is the global minimizer of $f(\vz)$. Furthermore, the function $g(\vu,\vv)$ has no spurious local minima, i.e., all local minimizers are global minimizers and each saddle point of $g(\vu,\vv)$ has a negative curvature direction.
\end{theorem}
\begin{proof}
Note that the function $g(\vu,\vv)$ is twice differentiable.  Therefore,  all local minimizers and saddle points that do not have a direction of negative curvature should obey
\begin{equation} \label{eq:locsad}
\nabla g(\vu,\vv) =\bm{0} \qquad \mbox{and} \qquad \nabla^2 g(\vu,\vv) \succeq \bm{0}.
\end{equation}
With this fact in place, we first show that for any point $(\hu,\hv) \in \R^n \times \R^n$ obeying \eqref{eq:locsad},  the point $\z:=\hu\circ\hv$ must be a global minimizer of $f(\vz)$. To this end, since $f(\vz)$ is a convex function, by the Karush-Kuhn-Tucher (KKT) conditions, we know that the vector $\z$ is a global minimizer of $f$ if the following holds:
\begin{subequations}
\begin{gather}
 \gamma(\z)_j =0,  \qquad  \mbox{for} \quad j\in T;  \label{eq:kkt1}\\
\gamma(\z)_j+ \mbox{sign}(\widehat{z}_j) \cdot \lambda =0,  \quad  \mbox{for} \quad j\in T^c, ~\widehat{z}_j\neq 0; \label{eq:kkt2}\\
\abs{ \gamma(\z)_j}   \le  \lambda , \qquad  \mbox{for} \quad j\in T^c, ~\widehat{z}_j= 0. \label{eq:kkt3}
\end{gather}
\end{subequations}
where $\gamma(\z)=A^\T (A \z - \vy) \in \R^n$.  For the function $g(\vu,\vv)$,  a  simple calculation gives the gradient
\begin{equation} \label{eq:grad}
\nabla g(\vu,\vv)=\left[
\begin{array}{l}
(A^\T A) \circ (\vv\vv^\T) \vu -(A^\T \vy)\circ \vv+\lambda\vu_{T^c}\\
(A^\T A) \circ (\vu\vu^\T) \vv -(A^\T \vy)\circ \vu+\lambda\vv_{T^c}
\end{array} \right]= \left[
\begin{array}{l}
\gamma(\vu,\vv) \circ \vv+\lambda\vu_{T^c}\\
\gamma(\vu,\vv) \circ \vu+\lambda\vv_{T^c}
\end{array} \right]
\end{equation}
and the Hessian matrix
\begin{equation} \label{eq:Hessian}
\nabla^2 g(\vu,\vv)=\left[
\begin{array}{ll}
(A^\T A) \circ (\vv\vv^\T) +\lambda \mbox{Diag}(\1_{T^c}) & (A^\T A) \circ (\vv\vu^\T) +  \mbox{Diag}(\gamma(\vu \circ \vv)) \\
(A^\T A) \circ (\vu\vv^\T) +  \mbox{Diag}(\gamma(\vu \circ \vv)) & (A^\T A) \circ (\vu\vu^\T) +\lambda \mbox{Diag}(\1_{T^c})
\end{array} \right],
\end{equation}
where $\gamma(\vu,\vv)=A^\T \xkh{A(\vu \circ \vv)- \vy}$.  Therefore, for any $(\hu,\hv)$ obeying \eqref{eq:locsad},  it comes from \eqref{eq:grad} that  for all $j\in T$, it holds
\begin{equation} \label{eq:gaulada}
\gamma(\hu,\hv)_j \widehat{v}_j =0, \quad  \quad \gamma(\hu,\hv)_j \widehat{u}_j =0,
\end{equation}
and for all $j\in T^c$, one has
\begin{equation} \label{eq:gaula}
\gamma(\hu,\hv)_j \widehat{v}_j+\lambda \widehat{u}_j=0, \quad \quad  \gamma(\hu,\hv)_j \widehat{u}_j+\lambda \widehat{v}_j=0.
\end{equation}
We next divide the proof into the following three cases:
\begin{itemize}
\item[(i)] $j\in T$ and $\widehat{u}_j \neq 0$ or $\widehat{v}_j \neq 0$.  From \eqref{eq:gaulada}, we have $ \gamma(\hu,\hv)_j =0$. Recall that $\z:=\hu\circ\hv$.  Therefore,  the equation \eqref{eq:kkt1} holds for the case where $\widehat{u}_j \neq 0$ or $\widehat{v}_j \neq 0$.
\item[(ii)]  $j\in T^c$ and $\widehat{u}_j \neq 0$ or $\widehat{v}_j \neq 0$. Due to the fact that $\lambda >0$, it  follows from  \eqref{eq:gaula} that for all $j\in T^c$, if $\widehat{u}_j  \neq 0$ then $\widehat{v}_j \neq 0$, and vise verse. Therefore, we have
\[
\gamma(\hu,\hv)_j = -\lambda \frac{ \widehat{u}_j}{ \widehat{v}_j}=-\lambda \frac{ \widehat{v}_j}{ \widehat{u}_j}.
\]
This gives that
\[
\abs{\widehat{u}_j}=\abs{\widehat{v}_j} \quad \mbox{and} \quad \gamma(\hu,\hv)_j = -\lambda \frac{\mbox{sign}( \widehat{u}_j)}{\mbox{sign}( \widehat{v}_j)}= -\lambda \mbox{sign}( \widehat{u}_j  \widehat{v}_j)= -\lambda  \mbox{sign}(\widehat{z}_j).
\]
And so the equation \eqref{eq:kkt2} holds.
\item[(iii)] $\widehat{u}_j = \widehat{v}_j = 0$.  Define the set $S_{\hu,\hv}=\supp(\hu)\cup \supp(\hv) \subset [n]$. Since $\nabla^2 g(\hu,\hv) \succeq \bm{0}$, it then comes from \eqref{eq:Hessian} that  the sub-matrix obeys
\[
\zkh{\nabla^2 g(\hu,\hv)}_{S_{\hu,\hv}^c, S_{\hu,\hv}^c}= \left[
\begin{array}{ll}
\lambda \mbox{Diag}(\1_{T^c \cap S_{\hu,\hv}^c}) &  \mbox{Diag}(\gamma(\vu \circ \vv)_{S_{\hu,\hv}^c}) \\
 \mbox{Diag}(\gamma(\vu \circ \vv)_{S_{\hu,\hv}^c}) & \lambda \mbox{Diag}(\1_{T^c \cap S_{\hu,\hv}^c})
\end{array} \right] \succeq \bm{0}.
\]
 A simple calculation gives that the eigenvalues of the sub-matrix are $\pm \gamma(\vu \circ \vv)_{j}$ when $j\in T\cap S_{\hu,\hv}^c$ and $\lambda \pm \gamma(\vu \circ \vv)_{j}$ when $j\in T^c \cap S_{\hu,\hv}^c$. Therefore, $\zkh{\nabla^2 g(\hu,\hv)}_{S_{\hu,\hv}^c, S_{\hu,\hv}^c} \succeq \bm{0}$ means that $\gamma(\vu \circ \vv)_{j}=0$ for all $j\in T$ with $\widehat{u}_j =\widehat{v}_j = 0$; and $ \abs{\gamma(\vu \circ \vv)_{j}} \le \lambda$ for all $j\in T^c$ with $\widehat{u}_j =\widehat{v}_j = 0$.
\end{itemize}
 In summary, we can see that if $(\hu,\hv)$ obeying \eqref{eq:locsad} then the point $\z:=\hu\circ\hv$ will obey \eqref{eq:kkt1}, \eqref{eq:kkt2}, and \eqref{eq:kkt3}, which means $\hu\circ\hv$ is a  global minimizer of $f$.

We next turn to prove the second part where the function $g(\vu,\vv)$ has no spurious local minima. To this end,  it suffices to show that  if $(\hu,\hv)$ obeying \eqref{eq:locsad} then it must be a global minimizer of $g(\vu,\vv)$. Since we have proved that if $(\hu,\hv)$ obeying \eqref{eq:locsad}, then $\hu \circ \hv$ is the global minimizer of $f$ and $\abs{\widehat{u}_j}=\abs{\widehat{v}_j}$ for all $j\in T^c$. With this in hand,
we have
\[
g(\hu,\hv)=f(\hu \circ \hv)=\inf_{\vz}~f(\vz)=\inf_{\vu,\vv}~g(\vu,\vv),
\]
where the last equality comes from Lemma \ref{le:sameinf}. This implies that $(\hu,\hv)$ a global minimizer of $g(\vu,\vv)$, which completes the proof.
\end{proof}

\subsection{A new proximal HPP algorithm for tail-Lasso}
As shown in Theorem \ref{th:geo}, the loss function $g(\vu,\vv)$ given in  \eqref{eq:hpp} has no spurious local minima. Therefore, the gradient method with random initialization will converge to a global minimizer of $g$ with high probability \cite{lee2019,jin2017}. However, the rate of convergence is unclear, and the gradient descent may take exponential time to escape saddle points \cite{du2017}. In this paper, we use the proximal alternating minimization method to solve \eqref{eq:hpp}. Compared to gradient-based algorithms, proximal alternating minimization not only converges quickly in practice, but is also easier to implement since the subproblems typically have closed-form solutions.  The proximal alternating minimization for solving $g(\vu,\vv)$ is summarized in Algorithm \ref{al:1}.

\begin{algorithm}[H]
\caption{Proximal Hadamard Product Parametrization (PHPP)}
\label{al:1}
\begin{algorithmic}[H]
\Require
Initial $(\vu_0,\vv_0) \in \R^n\times \R^n$; Estimated support $T\subset [n]$;
Parameter $\lambda, \alpha>0$; Maximum number of iterations $T_{\max}$.   \\
\textbf{for} $k=0,1,\ldots, T_{\max}$ \quad \textbf{do}
\begin{equation} \label{al1:1}
\vu_{k+1}= \mbox{argmin}_{\vu\in \R^n} \quad g(\vu,\vv_k)+\frac{1}{2\alpha} \norm{\vu-\vu_k}^2
\end{equation}
and
\begin{equation} \label{al1:2}
\vv_{k+1}=\mbox{argmin}_{\vv\in \R^n} \quad g(\vu_{k+1},\vv)+\frac{1}{2\alpha} \norm{\vv-\vv_k}^2
\end{equation}
\textbf{end for}
\Ensure
$\vx_{T_{\max}}:=\vu_{T_{\max}} \circ \vv_{T_{\max}}$.\\
\end{algorithmic}
\end{algorithm}

Referring to the definition of $g(\vu,\vv)$ given in \eqref{eq:hpp}, it is straightforward to observe that the minimizers of \eqref{al1:1} and \eqref{al1:2} admit closed-form expressions:
\begin{equation} \label{updateu}
\vu_{k+1}=\xkh{(A^\T A) \circ (\vv_k \vv_k^\T) + \lambda I_{T^c} + \frac1{\alpha} I_n }^{-1} \xkh{(A^\T \vy)\circ \vv_k+\frac1{\alpha} \vu_k}
\end{equation}
and
\begin{equation} \label{updatev}
\vv_{k+1}=\xkh{(A^\T A) \circ (\vu_{k+1} \vu_{k+1}^\T) + \lambda I_{T^c} + \frac1{\alpha} I_n }^{-1} \xkh{(A^\T \vy)\circ \vu_{k+1}+\frac1{\alpha} \vv_k}.
\end{equation}
Here, $I_n \in \R^{n\times n}$ is the identity matrix, and $I_{T^c}$ denotes a diagonal matrix where $(I_{T^c})_{j,j}=0$ for $j\in T$ and $(I_{T^c})_{j,j}=1$ for $j\in T^c$.  Although Algorithm \ref{al:1} is straightforward to implement, it requires two matrix inversions per iteration, making it increasingly computationally expensive as the dimension $n$ grows. However, by leveraging the structure of the PHPP, we can significantly reduce the computational burden through some modifications in sparse high-dimensional setting. Specifically,
denote $S_k=\supp(\vv_k)$. Then the update \eqref{updateu} can be written as
\begin{equation} \label{eq:upusk}
\vu_{k+1, S_k}=\xkh{(A_{S_k}^\T A_{S_k}) \circ (\vv_{k,S_k} \vv_{k,S_k}^\T) + \lambda I_{T^c \cap S_k} + \frac1{\alpha} (I_n)_{S_k} }^{-1} \xkh{(A_{S_k}^\T \vy)\circ \vv_{k,S_k}+\frac1{\alpha} \vu_{k,S_k}}
\end{equation}
and
\begin{equation} \label{eq:upuskc}
\vu_{k+1, j}= \left\{ \begin{array}{ll}  u_{k,j},   &  j\in T\cap S_k^c \vspace{1em}\\
   \frac{1}{\lambda\alpha+1} u_{k,j},  &  j\in T^c \cap S_k^c
   \end{array} \right. .
\end{equation}
 Here, $A_{S_k}$ is  the submatrix consisting of the columns of matrix  $A \in \R^{m\times n}$ indexed by $S_k$.  If the cardinality of $S_k$ is small, then PHPP update for $\vu_{k+1}$ can be computed efficiently. In practice, the set $S_k$ is set to be
\[
S_k=\dkh{j \in [n]:   |v_{k,j}|  \ge \tau },
\]
where $\tau>0$ is a predefined threshold.   Similarly, the update for $\vv_{k+1}$ is given by
\begin{equation} \label{eq:upuskv}
\vv_{k+1, \widetilde{S_k}}=\xkh{(A_{\widetilde{S_k}}^\T A_{\widetilde{S_k}}) \circ (\vu_{k+1,\widetilde{S_k}} \vu_{k+1,\widetilde{S_k}}^\T) + \lambda I_{T^c \cap \widetilde{S_k}} + \frac1{\alpha} (I_n)_{\widetilde{S_k}} }^{-1} \xkh{(A_{\widetilde{S_k}}^\T \vy)\circ \vu_{k+1,\widetilde{S_k}}+\frac1{\alpha} \vv_{k,\widetilde{S_k}}}
\end{equation}
and
\begin{equation} \label{eq:upuskcv}
\vv_{k+1, j}= \left\{ \begin{array}{ll}  v_{k,j},   &  j\in T\cap \widetilde{S_k^c} \vspace{1em}\\
   \frac{1}{\lambda\alpha+1} v_{k,j},  &  j\in T^c \cap \widetilde{S_k^c}
   \end{array} \right. .
\end{equation}
Here, $\widetilde{S_k}=\dkh{j \in [n]:   |u_{k+1,j}|  \ge \tau }$.  To improve the efficiency of PHPP in practice, the estimated support $T$ can be selected in an adaptive strategy where $T_0$ indexes the $k_0$ (normally $k_0=1$) largest components in magnitude of the initial point $\vz_0$, $T_1$ indexes the $k_0+k'$ (typically $k'=1$) largest  components of the second iteration $\vz_1$,  $T_2$ indexes the $k_0+2k'$ largest components of the third iteration $\vz_2$, and so on.  Based on the above analysis, the PHPP algorithm with a threshold $\tau>0$ is summarized in Algorithm \ref{al:2}. It is important to highlight that our algorithm does not require prior knowledge of the sparsity level $k$.

\begin{algorithm}[H]
\caption{Proximal Hadamard Product Parametrization: Improved Version}
\label{al:2}
\begin{algorithmic}[H]
\Require
Initial point $(\vu_0,\vv_0) \in \R^n\times \R^n$; Threshold $\tau>0$; Parameter $\lambda, \alpha, k_0, k'>0$; Maximum number of iterations $T_{\max}$.   \\
\textbf{Initialize:}  $\vz_0=\vu_0\circ \vv_0$, $k_T=k_0$ and $T_0=\supp\xkh{\mathbb{H}_{k_T}(\vz_0)}$.\\
\textbf{for} $k=0,1,\ldots, T_{\max}$,  \quad \textbf{do}
\begin{itemize}
\item[1.] Let $S_k=\dkh{j \in [n]:   |v_{k,j}|  \ge \tau }$.
\item[2.] Update $\vu_{k,S_k}$ according to \eqref{eq:upusk} with estimated support $T_k$.
\item[3.] Update $\vu_{k,S_k^c}$ according to \eqref{eq:upuskc} with estimated support $T_k$.
\item[4.] Let $\widetilde{S_k}=\dkh{j \in [n]:   |u_{k+1,j}|  \ge \tau }$.
\item[5.] Update $\vv_{k,\widetilde{S_k}}$ according to \eqref{eq:upuskv} with estimated support $T_k$.
\item[6.] Update $\vv_{k,\widetilde{S_k}^c}$ according to \eqref{eq:upuskcv} with estimated support $T_k$.
\item[7.] Set $\vz_k=\vu_k\circ \vv_k$ and $k_T=k_T+k'$. Update $T_{k+1}=\supp\xkh{\mathbb{H}_{k_T}(\vz_k)}$.
\end{itemize}
\textbf{end for}
\Ensure
$\vz_{T_{\max}}$.
\end{algorithmic}
\end{algorithm}

\subsection{Linear rate of convergence of PHPP}
In this section, we prove that Algorithm \ref{al:1} with random initialization converges linearly to a global solution of $g(\vu,\vv)$ with high probability.
We begin with a standard result demonstrating that the loss function $g(\vu_k,\vv_k)$ is decreasing, which plays a key role in the  convergence analysis.

\begin{lemma} \label{le:dest}
For any matrix $A \in \R^{m\times n}$, starting from any initial guess $(\vu_0,\vv_0)$, let $\dkh{(\vu_k,\vv_k)}_{k=1}^\infty$ are generated  by Algorithm \ref{al:1}.  Then, for all $k\ge 1$,  $(\vu_k,\vv_k)$ are well-defined and the following holds:
\[
g(\vu_{k+1},\vv_{k+1}) \le g(\vu_k,\vv_k)- \frac{1}{2\alpha} \norm{\vw_{k+1}-\vw_k}^2.
\]
Here, we denote $\vw_k:=(\vu_k,\vv_k)$ for all $k\ge 0$.
\end{lemma}
\begin{proof}
For any $\alpha>0$ and $\vv_k \in \R^n$, the function $\vu\to g(\vu,\vv_k)+\frac{1}{2\alpha} \norm{\vu-\vu_k}^2$ is lower bounded, strongly convex and coercive.  Therefore, $\vu_{k+1}$ is uniquely defined. Similarly, $\vv_{k+1}$ is also well-defined and unique.  By the definition,  one has
\[
g(\vu_{k+1},\vv_k)+\frac{1}{2\alpha} \norm{\vu_{k+1}-\vu_k}^2 \le g(\vu_k,\vv_k)
\]
and
\[
g(\vu_{k+1},\vv_{k+1})+\frac{1}{2\alpha} \norm{\vv_{k+1}-\vv_k}^2 \le g(\vu_{k+1},\vv_k).
\]
Combining the above two inequalities, we have
\[
g(\vu_{k+1},\vv_{k+1}) \le g(\vu_k,\vv_k)- \frac{1}{2\alpha} \norm{\vu_{k+1}-\vu_k}^2-\frac{1}{2\alpha} \norm{\vv_{k+1}-\vv_k}^2.
\]
This completes the proof.
\end{proof}

To establish the linear rate of convergence of the Algorithm \ref{al:1}, the tool of the Kurdyka-{\L}ojasiewicz (KL) inequality is needed.

\begin{definition}\cite{lojasiewicz,attouch2010} \label{def:KL}
Let $f:\R^n\to \R$ be a differentiable function, and $\bar\vx$ be a critical point of $f$. We say the function $f$ satisfies the Kurdyka-{\L}ojasiewicz (KL) property at point $\bar\vx$ if there exist $\eta\in (0,+\infty)$, a neighborhood $U$ of $\bar\vx$ and a continuous concave function $\psi: [0,\eta) \to \R_+$ with $\psi(0)=0$ such that
\begin{itemize}
\item[(i)] $\psi$ is continuously differentiable on $(0,\eta)$ with $\psi'(t) >0$.
\item[(ii)] for any $\vx \in U$ with $0<f(\vx) < f(\bar\vx)+\eta$, it holds that
\begin{equation} \label{eq:pkl}
\psi'\xkh{f(\vx)-f(\bar\vx)} \norm{\nabla f(\vx)}  \ge 1.
\end{equation}
\end{itemize}
\end{definition}

Assume that $\psi(t)= c t^{1-\theta}$ for some $c>0$ and $\theta\in (0,1)$. Then \eqref{eq:pkl} can be equivalently written as
\begin{equation}  \label{eq:KLdef}
c(1-\theta) \norm{\nabla f(\vx)} \ge \xkh{f(\vx)-f(\bar\vx)}^\theta.
\end{equation}
Usually, the smallest $\theta$ satisfying \eqref{eq:KLdef} is called the {\L}ojasiewicz exponent of $f$ at critical point $\bar\vx$. The KL property plays a key role in the convergence analysis of first-order methods (see \cite{attouch2010, attouch2013}), and the exponent is closely related to the local convergence rates of these algorithms \cite{attouch2009}.

The KL property is known to be satisfied for a large variety of functions.  In 1963, {\L}ojasiewicz proved that real-analytic functions satisfy the KL property in a critical point with $\theta \in [1/2,1)$ \cite{lojasiewicz}. Recently,  Kurdyka extended this result to differentiable functions with $o$-minimal structure \cite{kurdyka}.  The extension of KL inequality to nonsmooth subanalytic functions can be found in \cite{bolte2007,bolte20072}. For our proof, we require that the function $g(\vu,\vv)$ satisfies KL property at any global minimizer with exponent $\theta=1/2$. The result is the following proposition.

\begin{prop} \label{pro:kl}
For any matrix $A \in \R^{m\times n}$ with $\mbox{rank}(A)=m$.  Let $(\vu^*,\vv^*)$ be a global minimizer of $g(\vu,\vv)$ in \eqref{eq:hpp}, and let  $\vz^*=\vu^* \circ \vv^*$. Assume that $\norms{\vz^*}_0 \le m$ and $-A^\T \xkh{ A\vz^* -  \vy} \in \mathrm{ri}  \xkh{\lambda \norms{\vz_{T^c}^*}_{1}} $. If the estimated support $T$ obeying $T\subset \mbox{supp}(\vz^*)$,   then there exists constants $C, \epsilon>0$ and an exponent $\theta=1/2$ such that it holds
\[
\norm{\nabla g(\vu,\vv)} \ge C \abs{g(\vu,\vv)- g(\vu^*,\vv^*) }^\theta, \qquad \forall (\vu,\vv) \in U(\vu^*,\vv^*).
\]
Here,  $U(\vu^*,\vv^*) $ is a ball at $(\vu^*,\vv^*)$ with radius $\epsilon$, and $\mathrm{ri} \xkh{\cdot}$ denotes the relative interior set.
\end{prop}
\begin{remark}
Set $h(\vz):=\frac12\norm{A\vz-\vy}^2$. Then the functions $f$ and $g$ defined in \eqref{eq:orhpp} and \eqref{eq:hpp} can be rewritten as  $f(\vz)= h(\vz)+ \lambda\normone{\vz_{T^c}}$ and  $g(\vu,\vv)=h(\vu\circ \vv)+ \frac{\lambda}2 \xkh{\norm{\vu_{T^c}}^2+\norm{\vv_{T^c}}^2}$. To guarantee the KL property of $g$ with exponent $1/2$, as shown in Proposition \ref{pro:kl},  one requires the condition $-\nabla h(\vz^*) \in \mathrm{ri}  \xkh{\lambda \norms{\vz_{T^c}^*}_{1}} $. Such condition is also adopted in  \cite{Ouyang} for establishing KL property and  is referred to as the strict complementarity condition for $f$.  We conjecture that when the strict complementarity condition fails,  the result in Proposition \ref{pro:kl} does not hold. Another requirement is $T\subset \mbox{supp}(\vz^*)$. For the classical Lasso where $T=\emptyset$, it holds trivially. For our tail-Lasso, as we stated in Remark \ref{re:28},  if $T$ is selected using an adaptive strategy, then  after a few iterations, the set $T$ converges towards   $\mbox{supp}(\vz^*)$.
\end{remark}

\begin{proof}
Recall that the function  $g(\vu,\vv)$ is  a twice differentiable function. Let
\[
c_1=\lambda_{\mbox{min}} \xkh{\nabla^2 g(\vu^*,\vv^*)} \qquad \mbox{and} \qquad  c_2=\lambda_{\mbox{max}} \xkh{\nabla^2 g(\vu^*,\vv^*)}.
\]
Here, $\lambda_{\min}(\cdot)$ and $\lambda_{\max}(\cdot)$ denote the smallest and largest eigenvalues of a matrix.
We claim that the Hessian $\nabla^2 g(\vu^*,\vv^*)$ is positive definite, i.e., $c_1>0$. For convenience, denote $\vw=(\vu,\vv)$ and $\vw^*=(\vu^*,\vv^*) $. Since $\vw^*$ is a critical point of $g(\vz)$,  therefore  the Taylor expansion of $g$ at $\vw^*$ is
\[
g(\vw)- g(\vw^*) = (\vw-\vw^*)^\T \nabla^2 g(\vw^*) (\vw-\vw^*) +o(\norm{\vw-\vw^*}^2).
\]
It means that there exists $\epsilon_1>0$ such that for all $\vw$ obeying $\norm{\vw-\vw^*} \le \epsilon_1$, it holds
\begin{equation} \label{eq:KLlow}
\abs{g(\vw)- g(\vw^*)}  \le \xkh{c_2+1} \norm{\vw-\vw^*}^2.
\end{equation}
For the gradient $\nabla g(\vw)$, applying the Taylor expansion yields
\[
\nabla g(\vw) = \nabla^2 g(\vw^*) (\vw-\vw^*) + o(\norm{\vw-\vw^*}).
\]
Similarly, there exists $\epsilon_2>0$ such that all $\vw$ obeying $\norm{\vw-\vw^*} \le \epsilon_2$, it holds
\begin{eqnarray}
\norm{\nabla g(\vw) } & \ge &  \norm{ \nabla^2 g(\vw^*) (\vw-\vw^*) } - \frac{c_1}2 \norm{\vw-\vw^*} \notag \\
&\ge & \frac{c_1}2 \norm{\vw-\vw^*}, \label{eq:KLup}
\end{eqnarray}
where the last inequality comes from the fact that
\[
\norm{ \nabla^2 g(\vw^*) (\vw-\vw^*) } \ge  \lambda_{\mbox{min}} \xkh{\nabla^2 g(\vw^*)} \norm{\vw-\vw^*}.
\]
  Taking $\epsilon=\min(\epsilon_1,\epsilon_2)$ and  combining \eqref{eq:KLlow} and \eqref{eq:KLup}, we obtain that for  all $\vw$ obeying $\norm{\vw-\vw^*} \le \epsilon$, it holds
\[
\norm{\nabla g(\vw) } \ge C \abs{g(\vw)- g(\vw^*)}^{\frac12},
\]
where $C=\frac{c_1}{2\sqrt{c_2+1}}$. This gives the conclusion.

It remain to prove the claim that $c_1>0$. From \eqref{eq:Hessian}, the Hessian matrix of $g$ at $(\vu^*,\vv^*)$ is
\begin{equation*} \label{eq:Hessian1}
\nabla^2 g(\vu^*,\vv^*)=\left[
\begin{array}{cc}
(A^\T A) \circ (\vv^* \vv^{*\T})  & (A^\T A) \circ (\vv^* \vu^{*\T})  \\
(A^\T A) \circ (\vu^* \vv^{*\T}) & (A^\T A) \circ (\vu^* \vu^{*\T})
\end{array} \right]+  \left[
\begin{array}{cc}
\lambda \mbox{Diag}(\1_{T^c}) &   \mbox{Diag}(\gamma(\vu^* \circ \vv^*)) \\
 \mbox{Diag}(\gamma(\vu^* \circ \vv^*)) & \lambda \mbox{Diag}(\1_{T^c})
\end{array} \right].
\end{equation*}
Denote $\widehat S=\mbox{supp}(\vz^*)$.
Therefore we can rewritten $\nabla^2 g(\vu^*,\vv^*)$ into
\begin{eqnarray}
\nabla^2 g(\vu^*,\vv^*) & = &\underbrace{ \left[
\begin{array}{ll}
(A^\T A) \circ (\vv^* \vv^{*\T})  & (A^\T A) \circ (\vv^* \vu^{*\T})  \\
(A^\T A) \circ (\vu^* \vv^{*\T}) & (A^\T A) \circ (\vu^* \vu^{*\T})
\end{array} \right] }_{:=B_1} \label{eq:firgg2} \\
&& +  \underbrace{ \left[
\begin{array}{cc}
\lambda \mbox{Diag}(\1_{T^c \cap \hS}) &   \mbox{Diag}(\gamma(\vu^* \circ \vv^*)_{\hS}) \\
 \mbox{Diag}(\gamma(\vu^* \circ \vv^*)_{\hS}) & \lambda \mbox{Diag}(\1_{T^c\cap \hS})
\end{array} \right]}_{:=B_2} \notag \\
&&  + \underbrace{ \left[
\begin{array}{cc}
\lambda \mbox{Diag}(\1_{T^c \cap \hS^c}) &   \mbox{Diag}(\gamma(\vu^* \circ \vv^*)_{\hS^c}) \\
 \mbox{Diag}(\gamma(\vu^* \circ \vv^*)_{\hS^c}) & \lambda \mbox{Diag}(\1_{T^c\cap \hS^c})
\end{array} \right]}_{:=B_3} . \notag
\end{eqnarray}

{\bf Bound for $B_1$:}
Observe that $\widehat S \subset \mbox{supp}(\vu^*)$ and $\widehat S \subset  \mbox{supp}(\vv^*)$.
 For any $(\vx_1,\vx_2) \in \R^{n} \times \R^n$,  if $\vx_{1, \hS} \neq \bm 0$ or $\vx_{2,  \hS} \neq \bm 0$ then
\begin{eqnarray*}
(\vx_1,\vx_2) ^\T B_1 (\vx_1,\vx_2) & = &  (\vx_{1}  \circ \vv_{ \hS}^* )^\T A^\T A  (\vx_{1}  \circ \vv_{ \hS}^*)  \\
&&  + 2  (\vx_{1}  \circ \vv_{  \hS}^*)^\T A^\T A  (\vx_2 \circ \vu_{  \hS}^*)  \\
&&+ (\vx_2 \circ \vu_{  \hS}^*)^\T A^\T A  (\vx_2 \circ \vu_{  \hS}^*)\\
&>& 0,
\end{eqnarray*}
where the last inequality comes from the fact that $\mbox{rank}(A)=m$ and  $\norms{\vx_{1}  \circ \vv_{  \hS}^*}_0 \le m$, $\norms{\vx_{2}  \circ \vu_{  \hS}^*}_0 \le m$ by assumption $|\widehat S| \le m$.

{\bf Bound for $B_2$ and $B_3$:}
Since $(\vu^*,\vv^*)$ is a global solution to the program $\mbox{min}_{\vu,\vv} ~ g(\vu,\vv)$ and $\vz^*$ is a  interior point, as shown in the proof of Theorem \ref{th:geo} , it holds:
\begin{subequations}
\begin{gather}
 \gamma(\vu^* \circ \vv^*)_j =0,  \qquad  \mbox{for} \quad j\in T;  \label{eq:kk1}\\
\gamma(\vu^* \circ \vv^*)_j+ \mbox{sign}(u_j^* v_j^*) \cdot \lambda =0,  \quad  \mbox{for} \quad j\in T^c, ~z_j^* \neq 0; \label{eq:kk2}\\
\abs{ \gamma(\vu^* \circ \vv^*)_j}   <  \lambda , \qquad  \mbox{for} \quad j\in T^c, ~z_j^*= 0. \label{eq:kk3}
\end{gather}
\end{subequations}
It gives that  $B_2  \succeq 0 $  and  $B_3 \succeq 0 $. Furthermore,  for any $(\vx_1,\vx_2) \in \R^{n} \times \R^n$,  if $\vx_{1,  \hS^c} \neq \bm 0$ or $\vx_{2,  \hS^c} \neq \bm 0$,  then one has
\begin{eqnarray*}
(\vx_1,\vx_2) ^\T B_3 (\vx_1,\vx_2) & = & \lambda \norm{\vx_{1, T^c \cap \hS^c}}^2 + \lambda \norm{\vx_{2,  T^c \cap \hS^c}}^2 + 2 \vx_{1, \hS^c}^\T \mbox{Diag}(\gamma(\vu^* \circ \vv^*)_{\hS^c})  \vx_{2,  \hS^c} \\
& = & \lambda \norm{\vx_{1,  \hS^c}}^2 + \lambda \norm{\vx_{2,  \hS^c}}^2 + 2 \vx_{1, \hS^c}^\T \mbox{Diag}(\gamma(\vu^* \circ \vv^*)_{\hS^c})  \vx_{2,  \hS^c} \\
&> &  \lambda \norm{\vx_{1,  \hS^c}}^2 + \lambda \norm{\vx_{2,  \hS^c}}^2 -2\lambda \nj{ \abs{\vx_{1,  \hS^c}}, \abs{\vx_{2, \hS^c}} }  \\
&\ge & 0,
\end{eqnarray*}
where the second line follows from the the assumption $T\subset \widehat S$, the third line arises from \eqref{eq:kk3} and the last line comes from Cauchy-Schwarz inequality. Putting the estimators for $B_1,B_2$ and $B_3$ into \eqref{eq:firgg2}, we obtain that for any $\bm 0\neq (\vx_1,\vx_2) \in \R^{n} \times \R^n$, it holds
\[
(\vx_1,\vx_2) ^\T \nabla^2 g(\vu^*,\vv^*) (\vx_1,\vx_2)  >0,
\]
which implies $c_1=\lambda_{\mbox{min}} \xkh{\nabla^2 g(\vu^*,\vv^*)}>0$. This completes the proof.


\end{proof}

We are ready to establish the linear rate of convergence of the Algorithm \ref{al:1}.

\begin{theorem}
For any matrix $A \in \R^{m\times n}$ with $\mbox{rank}(A)=m$. Assume the estimated support $T$ obeying $|T| \le m$. For any initial point $(\vu_0,\vv_0) \in \R^n\times \R^n$,  let $\dkh{(\vu_k,\vv_k)}_k$ be the sequence generated by Algorithm \ref{al:1} with a suitable parameter $\alpha >0$. Then
\begin{itemize}
\item[(i)] The sequence $\dkh{(\vu_k, \vv_k)}_{k=0}^{\infty}$ converges to a critical point $(\vu^*,\vv^*)$ of the function $g(\vu,\vv)$.
\item[(ii)] The critical point $(\vu^*,\vv^*)$ is a global solution of $g(\vu,\vv)$ almost surely. In other words, the set of  $(\vu_0,\vv_0)$ such that Algorithm \ref{al:1} does not converge to a global solution is of zero Lebesgue measure.
\item[(iii)]  Denote $\vz^*=\vu^* \circ \vv^*$. In the case $\mathrm{(ii)}$,  assume in addition that  $T\subset \mbox{supp}(\vz^*)$,  $\norms{\vz^*}_0 \le m$,  and  $-A^\T \xkh{ A\vz^* -  \vy} \in \mathrm{ri}  \xkh{\lambda \norms{\vz_{T^c}^*}_{1}}$. Then there exist $c>0$,  $\tau \in (0,1)$ and an integer $k_0$ such that for all $k\ge k_0$, it holds
\[
\norm{(\vu_k,\vv_k)-(\vu^*,\vv^*)} \le c\tau^{k-k_0}.
\]
\end{itemize}
\end{theorem}

\begin{proof}
We first prove (i).  For convenience, we denote $\vw_k=(\vu_k,\vv_k)$ for all $k\ge 0$. By lemma \ref{le:dest}, we have
\begin{equation} \label{eq:dew}
 \frac{1}{2\alpha} \norm{\vw_{k+1}-\vw_k}^2 \le g(\vw_k)- g(\vw_{k+1}) .
\end{equation}
That is, $g(\vw_k) \ge 0$ is monotonically nonincreasing sequence, and its limit exists. Define the level set
\[
\mathcal{R}:=\dkh{\vw\in \R^n\times \R^n: g(\vw) \le g(\vw_0)}.
\]
Under the condition that  $\mbox{rank}(A)=m$ and $|T| \le m$, it is easy to verify that $g(\vw)$ is coercive, i.e., $g(\vw) \to \infty$ as $\norm{\vw} \to \infty$. Therefore, $\dkh{(\vu_k, \vv_k)}_k \subset \mathcal{R}$ is a bounded sequence, which means the set of its limit points is non-empty.
By the definition of $\vu_{k+1}$ and $\vv_{k+1}$, it follows from the optimality conditions that
\begin{equation} \label{eq:opc1}
\nabla_{\vu} g(\vu_{k+1}, \vv_k) + \frac1{\alpha} (\vu_{k+1}-\vu_k) =0
\end{equation}
and
\begin{equation} \label{eq:opc2}
\nabla_{\vv} g(\vu_{k+1}, \vv_{k+1}) + \frac1{\alpha} (\vv_{k+1}-\vv_k) =0.
\end{equation}
Summing \eqref{eq:dew} over $k$ from $0$ to $+\infty$,  and combining \eqref{eq:opc1} and \eqref{eq:opc2},
we obtain
\begin{eqnarray*}
\frac{\alpha}2  \sum_{k=0}^{+\infty} \xkh{\norm{\nabla_{\vu} g(\vu_{k+1}, \vv_k)}^2+\norm{ \nabla_{\vv} g(\vu_{k+1}, \vv_{k+1}) }^2 }&=& \frac{1}{2\alpha} \sum_{k=0}^{+\infty} \norm{\vw_{k+1}-\vw_k}^2\\
& \le & g(\vw_k)- \lim\limits_{k\to +\infty}g(\vw_{k}).
\end{eqnarray*}
This means
\[
 \lim\limits_{k\to +\infty} \norm{\vu_{k+1}-\vu_k}=\alpha \cdot \lim\limits_{k\to +\infty} \norm{\nabla_{\vu} g(\vu_{k+1}, \vv_k)}=0
\]
and
\[
 \lim\limits_{k\to +\infty} \norm{\vv_{k+1}-\vv_k}=\alpha \cdot \lim\limits_{k\to +\infty} \norm{\nabla_{\vv} g(\vu_{k+1}, \vv_{k+1})}=0.
\]
By the continuity of $\nabla_{\vu} g$ and $\nabla_{\vv} g$, we obtain that any clustering point of $\dkh{(\vu_k, \vv_k)}_k$ is a critical point of $g(\vu,\vv)$.

It remains to prove that $\dkh{(\vu_k, \vv_k)}_k$ is convergent, which can be checked by showing that $\dkh{(\vu_k, \vv_k)}_k$ is a Cauchy sequence. To this end,  for any clustering point $(\vu^*,\vv^*)$ of $\dkh{(\vu_k, \vv_k)}_k$,  we denote $\vw^*:=(\vu^*,\vv^*)$.
 Since $g(\vu,\vv)$ is a real-valued polynomial function, it then satisfies KL property.  From Definition \ref{def:KL}, there exist a differentiable concave function $\psi(t): \R_+ \to \R_+$ and a constant $\delta>0$ such that
\begin{equation} \label{eq:klw}
\psi'\xkh{g(\vw)-g(\vw^*)} \norm{\nabla g(\vw)} \ge 1 \qquad \mbox{for all} \quad  \vw \in B(\vw^*, \delta).
\end{equation}
Here, $B(\vw^*,\delta)$ is a neighborhood of $\vw^*$ with radius $\delta$. Note that $\vw^*$ is a clustering point of $\dkh{\vw_k}_{k\ge 1} $.  There is an integer $k_0$ such that $\vw_{k_0} \in B(\vw^*, \delta/2)$.  To show \eqref{eq:klw} holds for all $\vw_k$ with $k\ge k_0$, we next prove by induction that $\vw_k \in B(\vw^*,\delta)$ for all $k\ge k_0$.
Assume that $\vw_k \in B(\vw^*,\delta)$ holds up to some $k\ge k_0$.   Recall that $\psi(t)$ is a concave function, it then comes from \eqref{eq:klw} that
\begin{eqnarray}
&& \psi\xkh{g(\vw_k)-g(\vw^*) }- \psi\xkh{g(\vw_{k+1})-g(\vw^*) }  \notag \\
&\ge & \psi'\xkh{g(\vw_k)-g(\vw^*)} \cdot \xkh{g(\vw_k)- g(\vw_{k+1})}  \notag \\
&\ge & \frac{\norm{\vw_{k+1}-\vw_k}^2}{2\alpha\norm{\nabla g(\vw_k)}}, \label{eq:psminu}
\end{eqnarray}
where we use \eqref{eq:dew} in the last inequality. For the term $\norm{\nabla g(\vw_k)}$, by \eqref{eq:opc2}, we have
\begin{equation} \label{eq:grgu1}
\norm{\nabla_{\vv} g(\vu_{k}, \vv_{k}) } = \frac1{\alpha} \norm{\vv_{k}-\vv_{k-1}}.
\end{equation}
Furthermore,  using \eqref{eq:opc1}, we obtain
\begin{eqnarray}
\norm{\nabla_{\vu} g(\vu_{k}, \vv_{k}) } & = &  \norm{\nabla_{\vu} g(\vu_{k}, \vv_{k}) - \nabla_{\vu} g(\vu_{k}, \vv_{k-1}) -\frac1{\alpha} (\vu_{k}-\vu_{k-1})}  \notag \\
&\le &  \norm{\nabla_{\vu} g(\vu_{k}, \vv_{k}) - \nabla_{\vu} g(\vu_{k}, \vv_{k-1}) }+  \frac1{\alpha}  \norm{\vu_{k}-\vu_{k-1}} \notag  \\
&\le & L_g  \norm{\vv_{k}-\vv_{k-1}} + \frac1{\alpha}  \norm{\vu_{k}-\vu_{k-1}} \notag \\
&\le & C_1 \norm{\vw_k-\vw_{k-1}} \label{eq:grgu}
\end{eqnarray}
where $L_g:= \sup_{\vu,\vv \in \mathcal R} \norm{\nabla_{\vu,\vv}^2 g(\vu,\vv)}$ is a finite due to the fact $\nabla_{\vu,\vv}^2 g(\vu,\vv)$ is continuous function and the region $\mathcal R$ is a compact, and $C_1:=2 \max\dkh{L_g, 1/\alpha}$. Putting \eqref{eq:grgu1} and \eqref{eq:grgu} into \eqref{eq:psminu} gives
\[
\psi\xkh{g(\vw_k)-g(\vw^*) }- \psi\xkh{g(\vw_{k+1})-g(\vw^*) }  \ge \frac1{2\alpha C_1}\cdot \frac{\norm{\vw_{k+1}-\vw_k}^2}{ \norm{\vw_k-\vw_{k-1}}},
\]
which means
\begin{eqnarray}
 2 \norm{\vw_{k+1}-\vw_k} & \le &  2 \sqrt{\norm{\vw_k-\vw_{k-1}}} \sqrt{ 2\alpha C_1 \cdot \xkh{\psi\xkh{g(\vw_k)-g(\vw^*) }- \psi\xkh{g(\vw_{k+1})-g(\vw^*) } } } \notag \\
 &\le & \norm{\vw_k-\vw_{k-1}} + C \cdot \xkh{\psi\xkh{g(\vw_k)-g(\vw^*) }- \psi\xkh{g(\vw_{k+1})-g(\vw^*) } } . \label{eq:wkpwst}
\end{eqnarray}
Here, $C=2\alpha C_1>0$ is a constant. Summing  it up to $k$, we get
\begin{equation} \label{eq:S_0}
 \sum_{j=k_0}^k \norm{\vw_{j+1}-\vw_j} \le \norm{\vw_{k_0}-\vw_{k_0-1}} +C \psi\xkh{g(\vw_{k_0})-g(\vw^*) }.
\end{equation}
Therefore, it holds
\begin{eqnarray*}
\norm{\vw_{k+1}-\vw^*} & \le &  \sum_{j=k_0}^k \norm{\vw_{j+1}-\vw_j} +\norm{\vw_1-\vw^*}  \\
&\le &\norm{\vw_{k_0}-\vw_{k_0-1}} +C \psi\xkh{g(\vw_{k_0})-g(\vw^*) }+\norm{\vw_{k_0}-\vw^*} \\
&\le &2\alpha \xkh{g(\vw_{k_0})-g(\vw_{k_0-1})}+ C \psi\xkh{g(\vw_{k_0})-g(\vw^*) }+\norm{\vw_{k_0}-\vw^*} \\
& \le &  \delta,
\end{eqnarray*}
where the third inequality comes from \eqref{eq:dew}, and the  the last inequality comes from the fact that $\norm{\vw_{k_0}-\vw^*}  \le \delta/2$, and $\lim_{k\to \infty} g(\vw_k) =g(\vw^*)$. Hence, $\vw_{k+1} \in B(\vw^*,\delta)$.
This completes the proof that $\vw_{k} \in B(\vw^*,\delta)$ for all $k\ge k_0$. With this result in place, the inequality \eqref{eq:wkpwst} holds for all $k\ge k_0$. Summing it over $k$ from $k_0$ to $+\infty$, we get
\[
 \sum_{j=k_0}^{+\infty} \norm{\vw_{j+1}-\vw_j} \le \norm{\vw_{k_0}-\vw_{k_0-1}} +C \psi\xkh{g(\vw_{k_0})-g(\vw^*) }<+\infty.
\]
This verifies that $\dkh{(\vu_k, \vv_k)}_k$ is a Cauchy sequence, and therefore it is convergent.

We next turn to prove  (ii). As shown in \cite[Theorem 10]{li2019alternating}, if $g(\vu,\vv)$ is $L_g$ bi-smooth in the level set $\mathcal{R}:=\dkh{\vw\in \R^n\times \R^n: g(\vw) \le g(\vw_0)} $, i.e.,
 $\max\dkh{ \norm{\nabla_{\vu}^2 g(\vu,\vv)},  \norm{\nabla_{\vv}^2 g(\vu,\vv)} } \le L_g$ in $\mathcal R$, and the parameter $\alpha> L_g$, then Algorithm \ref{al:1} with random initialization will not converge to a strict saddle point of $g$ almost surely.  Observe that
\[
\nabla_{\vu}^2 g(\vu,\vv)=(A^\T A) \circ (\vv\vv^\T) +\lambda \mbox{Diag}(\1_{T^c}).
\]
Therefore, for any $\vw=(\vu,\vv) \in \mathcal R$, it holds
\[
\norm{\nabla_{\vu}^2 g(\vu,\vv)} \le \norm{A}^2 \norm{\vv}^2 + \lambda \le L_g
\]
 for some constant $L_g>0$. Here, the last inequality comes from the fact that $R$ is a bounded region such that $\norm{\vv} \le C'$ for some constant $C'>0$. Similarly, one can show that $\norm{\nabla_{\vv}^2 g(\vu,\vv)}  \le L_g$.  Therefore, the critical point $(\vu^*,\vv^*)$ generated by Algorithm \ref{al:1} is not a strict saddle point almost surely by \cite[Theorem 10]{li2019alternating}. Finally, due to the benign geometric landscape of $g(\vu,\vv)$ given in Theorem \ref{th:geo}, namely, all local minimizers are global minimizers and each saddle point are strict,  we arrive at the conclusion that $(\vu^*,\vv^*)$ is a global solution of $g(\vu,\vv)$ almost surely.

Finally, let us prove (iii).  In the case of (ii), $(\vu^*,\vv^*)$ is a global solution of $g(\vu,\vv)$ almost surely.  Applying  Proposition \ref{pro:kl}, the function $g(\vu,\vv)$ satisfies KL property at $(\vu^*,\vv^*)$ with a differentiable concave function $\psi(t)=c t^{1/2}$. Here, $c>0$ is a constant. Therefore, it comes from \eqref{eq:wkpwst} that
\[
 2 \norm{\vw_{k+1}-\vw_k} \le \norm{\vw_k-\vw_{k-1}} + C_2 \cdot \xkh{\sqrt{g(\vw_k)-g(\vw^*) }- \sqrt{g(\vw_{k+1})-g(\vw^*) } },
\]
where $C_2>0$ is a constant. For any $k\ge k_0$, let $S_k=\sum_{j=k}^{\infty} \norm{\vw_{k+1}-\vw_k}$.  Summing the above inequality over $k$, we get
\begin{equation} \label{eq:cc1}
\sum_{j=k}^\infty \norm{\vw_{j+1}-\vw_j} \le \norm{\vw_{k}-\vw_{k-1}} +C_2 \sqrt{g(\vw_{k})-g(\vw^*) }.
\end{equation}
Observing that $\vw_k$ satisfies the KL inequality \eqref{eq:klw}, it yields
\begin{equation} \label{eq:cc2}
\norm{g(\vw_k)} \ge \frac{1}{\psi'(g(\vw_k)-g(\vw^*))} =\frac{2}{c} \sqrt{g(\vw_{k})-g(\vw^*) },
\end{equation}
where the last inequality due to the fact $\psi'(t)=c/(2\sqrt t)$. Moreover, as shown in \eqref{eq:grgu1} and \eqref{eq:grgu}, one has
\begin{equation} \label{eq:cc3}
\norm{g(\vw_k)} \le 2C_1 \norm{\vw_k-\vw_{k-1}}.
\end{equation}
Combining \eqref{eq:cc1}, \eqref{eq:cc2} and \eqref{eq:cc3} together, and using the definition of $S_k$, we obtain
\[
S_k \le C_3 \xkh{S_{k-1} -S_k},
\]
where $C_3:=c\cdot C_1\cdot C_2>0$ is a constant. This gives
\[
S_k=\frac{C_3}{C_3+1} S_{k-1} \le \cdots \le \tau^{k-k_0} S_{k_0}
\]
for $\tau=c_3/(C_3+1)$. Since $\norm{\vw_k-\vw^*} \le S_k$ and $S_{k_0}$ is finite number by \eqref{eq:S_0}, we completes the proof.
\end{proof}

\section{Numerical experiments}
In this section, we present several numerical experiments to validate the effectiveness and robustness of the PHPP algorithm (Algorithm \ref{al:2})  in comparison with  HPP\cite{hoff2017}, Tail-HPP \cite{Guangxi}, OMP\cite{zhang2011sparse}, CoSaMP\cite{needell2009cosamp}, SP \cite{Dai}, and HTP\cite{foucart2011hard}. These methods are selected due to their widespread applications and high efficiency in solving compressed sensing.    The implementations for OMP, CoSaMP, SP, and HTP were obtained from Matlab central file exchange \footnote{OMP, CoSaMP, SP are available at: https://www.mathworks.com/matlabcentral/fileexchange/160128-omp-sp-cosamp. HTP is available at:  https://github.com/foucart/HTP.}. All experiments are conducted on a laptop with a 2.4 GHz Intel Core i7 Processor, 8 GB 2133 MHz LPDDR3 memory, and Matlab R2024a.

Throughout the numerical experiments, the target vector  $\vxs \in \R^n$ is set to be $k$-sparse, with its support uniformly drawn at random from all $k$-subsets of $\dkh{1,2,\ldots,n}$. The nonzero entries of  $\vx$ are randomly generated from the standard Gaussian distribution.  The measurement matrix $A$ is selected according to one of the following configurations, as described in \cite{Yin}.
\begin{itemize}
\item[1.] {\bf Gaussian matrix:} Let $A\in \R^{m\times n}$ be a random Gaussian matrix, where each column is independently generated from the standard normal distribution. The columns are then normalized to have unit length.
\item[2.] {\bf Partial DCT matrix:} Let $A\in \R^{m\times n}$ be a random partial discrete cosine transform matrix generated by
\[
A_{i,j}=\cos\xkh{2\pi(j-1)\psi_i}, \quad i=1,\ldots,m,\quad j=1,\ldots,n,
\]
where $\psi_i, i=1,\ldots,m$ are independently and uniformly sampled from $[0,1]$.  The columns of $A$ are then normalized to have unit length.
\end{itemize}

\begin{example}
In this example, we evaluate the empirical successful rate of PHPP in comparison with state-of-the-art methods in the absence of noise. We first conduct $500$ independent trials with fixed $n=256$ and $m=64$,  recording the successful rates for sparsity levels $k$ ranging from $1$ to $34$. A trial is considered successful if the algorithm returns a vector $\vz_T$  satisfying  $\norm{\vz_{T}-\vxs}\le 10^{-4}$.  The parameters of PHPP, as specified in Algorithm \ref{al:2}, are set to $\lambda=0.1/m$, $\alpha=10/\lambda$, and $\tau=10^{-5}$.  The results are presented in Figure \ref{figure:suc_fixm}, where it can be observed that PHPP achieves the highest success rate for each sparsity level $k$.  Next, we perform $500$ independent trials with fixed $n=256$ and  $k=12$, but  varying $m$ within the range $[0.1n,0.3n]$.  The results, shown in Figure \ref{figure:suc_fixk}, clearly demonstrate that as $m$   increases, the problem becomes easier to solve. Once again, PHPP outperforms the other methods, achieving the highest success rate across all sparsity $k$.
\end{example}

\begin{figure}[H]
\centering
\subfigure[]{
     \includegraphics[width=0.45\textwidth]{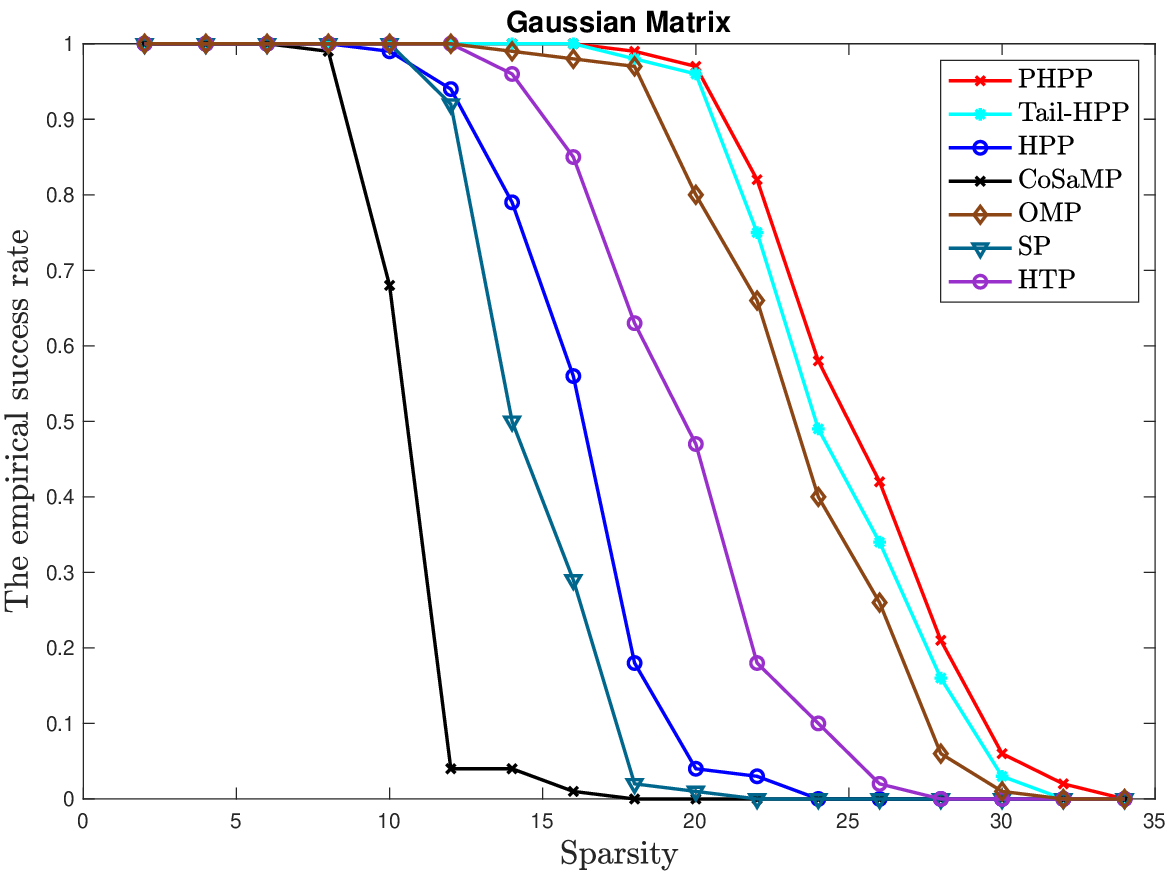}}
\subfigure[]{
     \includegraphics[width=0.45\textwidth]{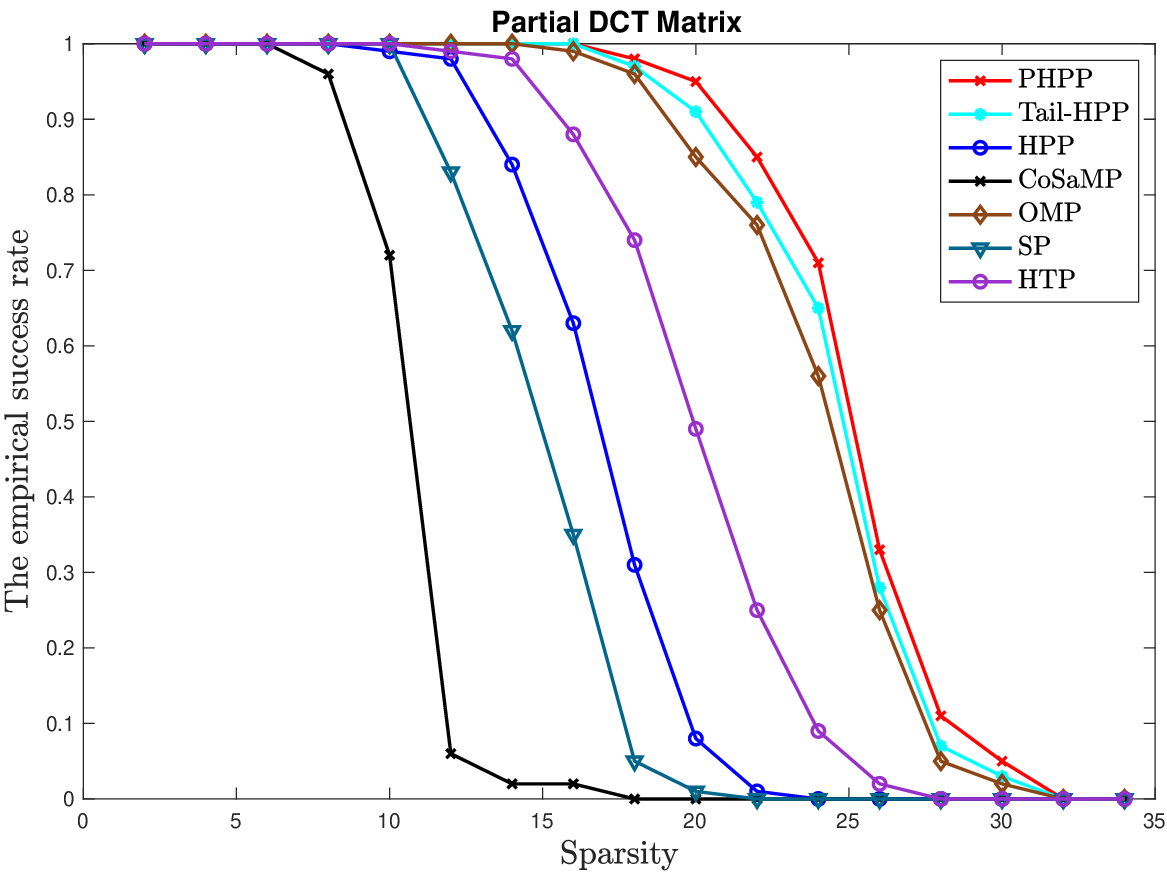}}
\caption{Success rates. $n=256, m=64$ and sparsity $k\in \dkh{2,4,\cdots,34}$.}
\label{figure:suc_fixm}
\end{figure}

\begin{figure}[H]
\centering
\subfigure[]{
     \includegraphics[width=0.45\textwidth]{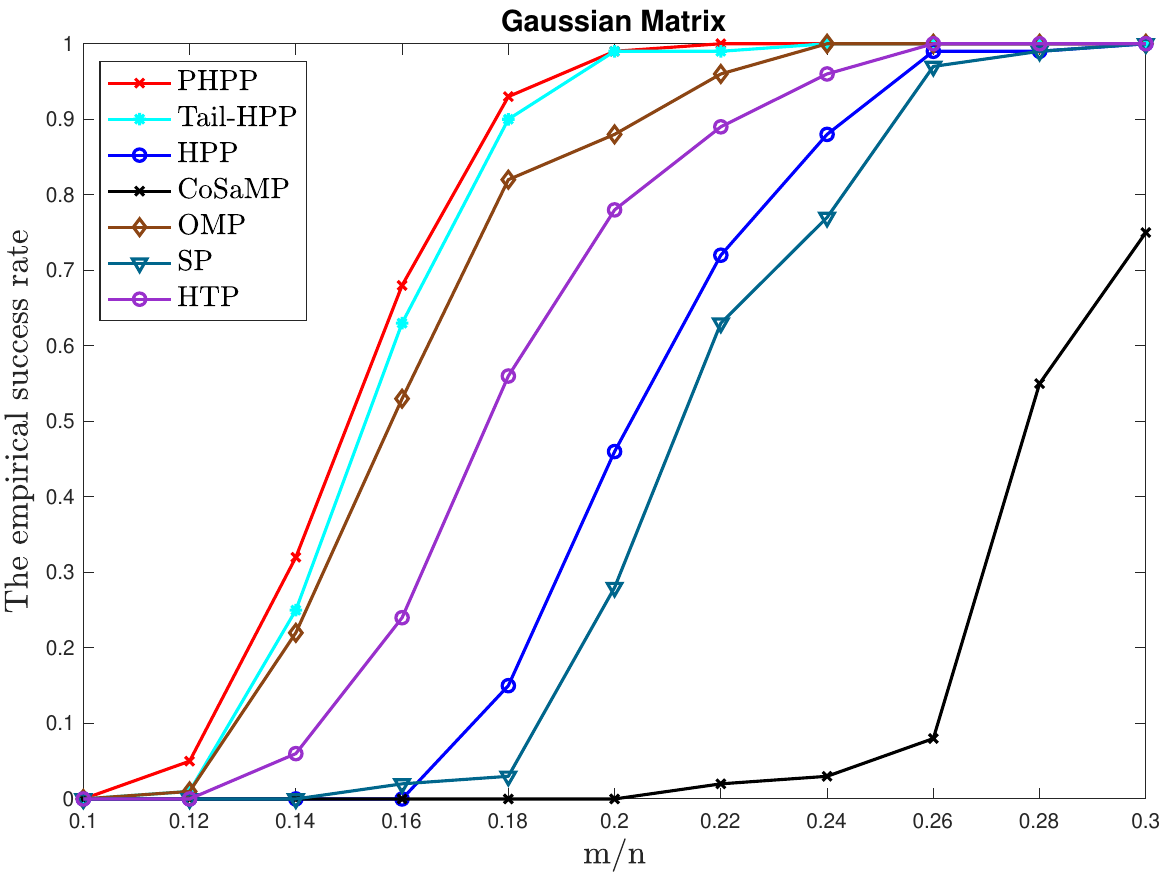}}
\subfigure[]{
     \includegraphics[width=0.45\textwidth]{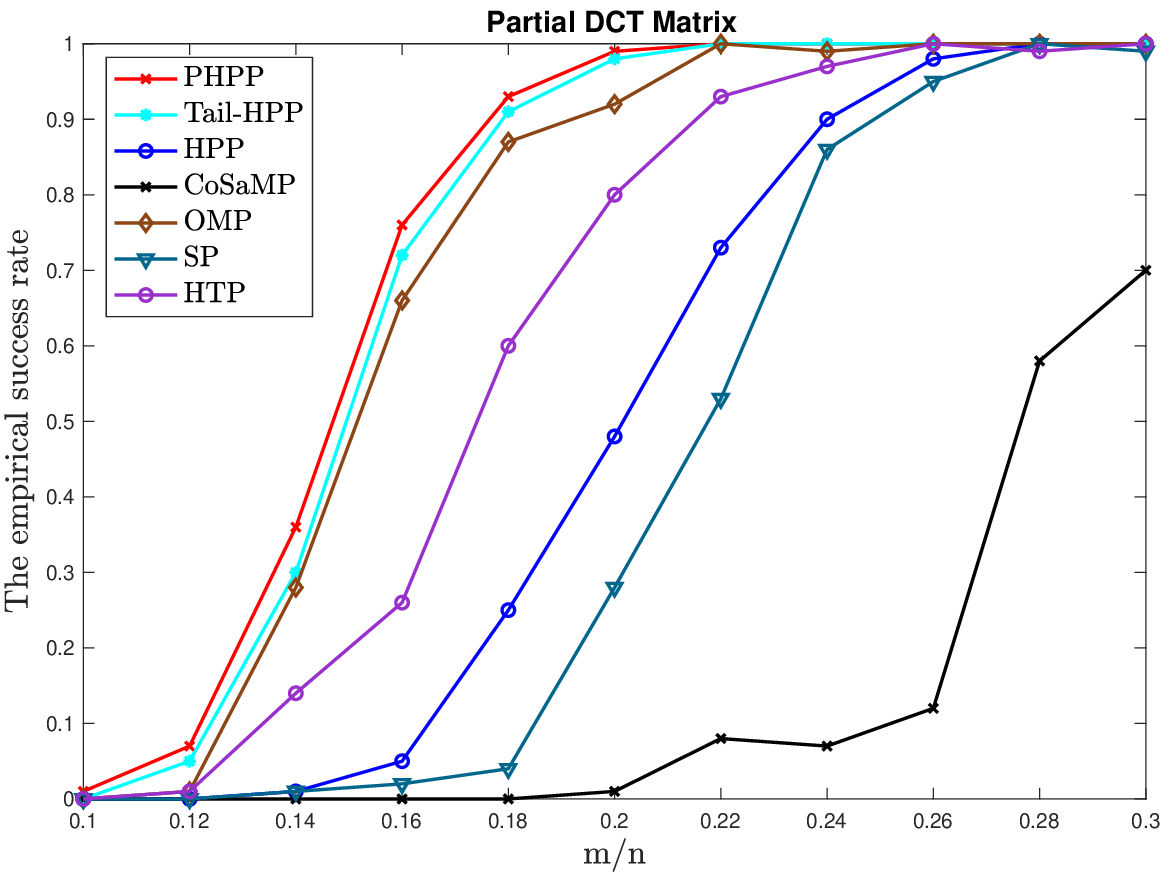}}
\caption{Success rates. $n=256, k=12$ and $m=0.1n, 0.12n, \ldots, 0.3n$.}
\label{figure:suc_fixk}
\end{figure}

\begin{example}
In this example, we compare the empirical successful rate of PHPP with those of state-of-the-art methods under noisy measurements $y=A\vx+\vw$, where $\vw \in \R^m$ is the noise vector.  Two types of noise distributions are considered. The first is Gaussian noise, where $\vw \sim \sigma\cdot \mathcal{N}(0, I_m)/\sqrt{m}$ with $\sigma=0.01$.  The second is heavy-tailed noise, where each entry of $\vw$ is generated according to Student's $t$-distribution with parameter $\mu=5$, specifically $w_j \sim  \sigma \cdot t_5/\sqrt{m}$ for all $j=1,\ldots,m$.   We run $500$ independent trials with fixed $n=256, m=64$ and record the success rates for sparsity levels $k$ varying from $1$ to $34$. In the noisy case, a trial is considered successful if the algorithm returns a vector $\vz_T$  such that  $\norm{\vz_{T}-\vxs}\le 0.01$.  The parameters for PHPP are set as $\lambda=\sigma\sqrt{\frac{\log n}{m}}$, $\alpha=10/\lambda$, and $\tau=10^{-5}$. The results are shown in Figure \ref{figure:suc_nois}.  It can be observed that  PHPP yields the highest success rate for each sparsity level $k$.
\end{example}

\begin{figure}[H]
\centering
\subfigure[]{
     \includegraphics[width=0.45\textwidth]{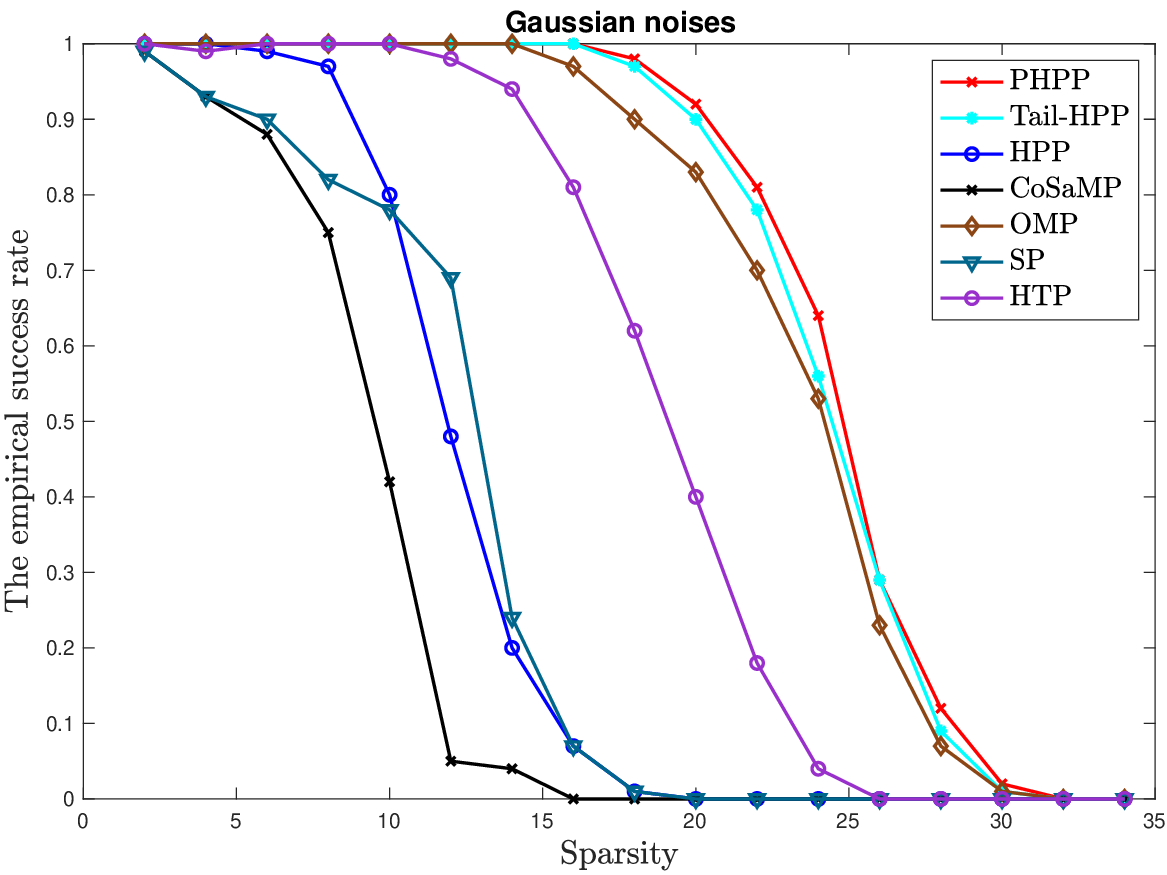}}
\subfigure[]{
     \includegraphics[width=0.45\textwidth]{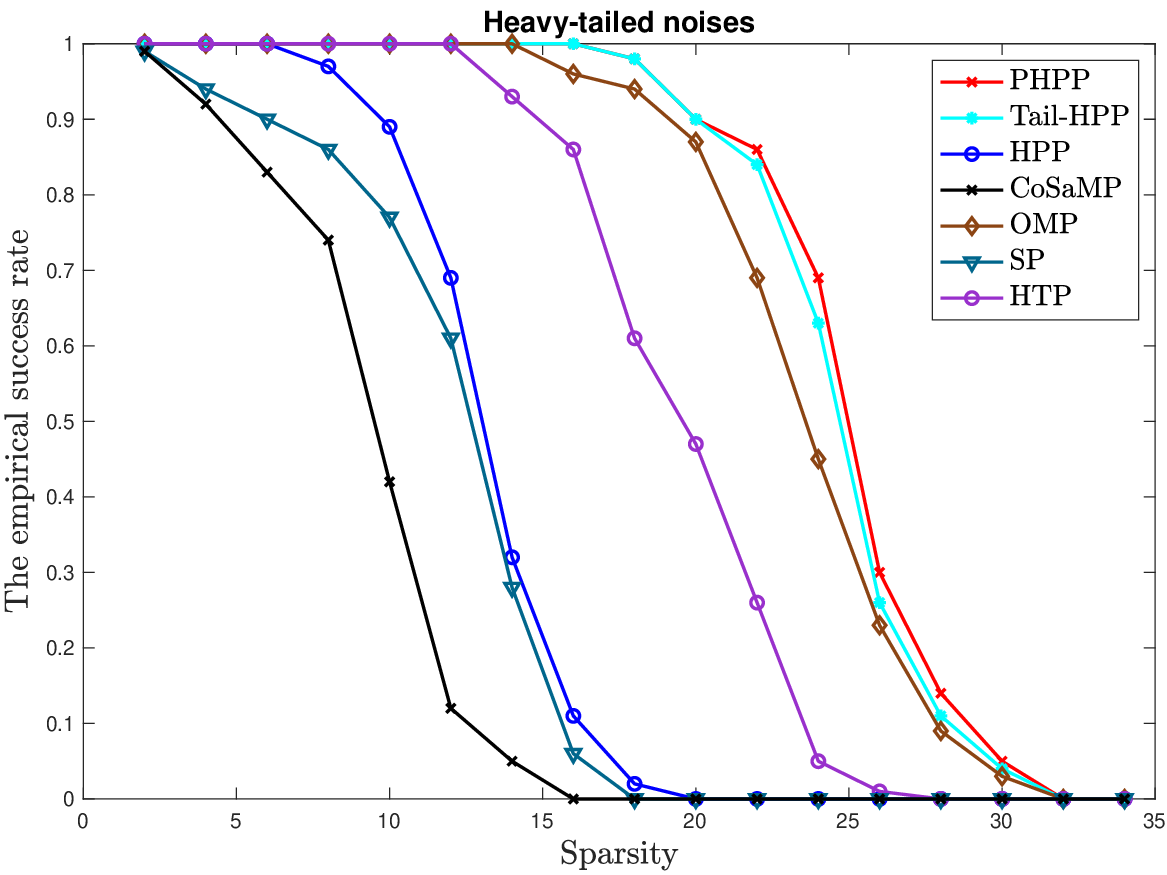}}
\caption{Success rates for noisy measurements. $n=256, m=64$ and sparsity $k\in \dkh{2,4,\cdots,34}$.}
\label{figure:suc_nois}
\end{figure}

\begin{table}[tp]
  \centering
  \fontsize{13}{16}\selectfont
  \caption{Average CPU time (in seconds) for PHPP with those of state-of-the-art methods.}
  \label{tab:computime}
    \begin{tabular}{|c|c|c|c|c|c|c|c|}
    \hline
     $n$ & \mbox{PHPP} &\mbox{Tail-HPP} & \mbox{HPP} & \mbox{CoSaMP} & \mbox{OMP}& \mbox{SP} & \mbox{HTP} \cr\hline
    1000 & 0.0201 & 0.0531 & 0.0732  & 0.0028 & 0.0043& 0.0016 & 0.0047  \cr \hline
    2000 &0.0937 & 0.2639 & 0.3473  & 0.0108 & 0.0087& 0.0049 & 0.0252  \cr \hline
    3000 & 0.2338 & 0.7429 & 1.0325  & 0.0278 & 0.0316 & 0.0109 & 0.0665 \cr \hline
    \end{tabular}
\end{table}

\begin{example}
In this example, we evaluate the computational time of these seven methods.  We run $50$ trials with higher dimensions $n=1000,2000,3000$, while keeping $m=n/4$ and the sparsity level $k=0.01n$.   The parameters for PHPP are set to $\lambda=0.1/m$, $\alpha=10/\lambda$, and $\tau=0.01$.  Table \ref{tab:computime} lists the averaged CPU time required to achieve an error of  $10^{-4}$. For our algorithm PHPP, during the first few iterations (approximately 4 iterations), the cardinality of $S_k$ and $\widetilde{S_k}$  in Algorithm \ref{al:2} are relatively large, making the computations in \eqref{eq:upusk} and \eqref{eq:upuskv} more expensive.

As shown in Table \ref{tab:computime}, the PHPP method is slower in overall runtime compared to state-of-the-art methods such as CoSaMP, OMP, SP, and HTP, as it involves iterations over estimated support $T_k$.  It is clear, however, PHPP achieves much higher successful rate of signal recovery. 
\end{example}

\section{Discussions}
This article addresses RIP analyses of the tail minimization problem for compressed sensing.  Recovery error bounds for both tail-$\ell_1$ minimization and tail-lasso problems under RIP conditions are established. Additionally, an improved version of the HPP method - proximal alternating minimization using Hadamard product parametrization - is proposed for the effective solution of the tail-lasso problem. The linear rate of convergence of the proposed algorithm is established using the Kurdyka-{\L}ojasiewicz  inequality.  Numerical results demonstrate that the proposed algorithm significantly enhances signal recovery performance compared to state-of-the-art techniques.

Several intriguing problems remain for future studies. First, while we have demonstrated the linear rate of convergence of the PHPP method for a given estimated support $T$, proving the convergence for adaptively estimated support $T$ would be an interesting challenge. Second, although the PHPP method achieves the highest success rate compared to state-of-the-art methods, continuing to reduce  the computational cost would be of significant practical interest.


\begin{thebibliography}{10}



\bibitem{attouch2009}
H. Attouch and J. Bolte,
\newblock ``On the convergence of the proximal algorithm for nonsmooth functions involving analytic features,''
\newblock {\em Math. Program.}, vol. 116, pp. 5--16, 2009.





\bibitem{attouch2010}
H. Attouch, J. Bolte, P. Redont, and A. Soubeyran,
\newblock ``Proximal alternating minimization and projection methods for nonconvex problems: An approach based on the Kurdyka-{\L}ojasiewicz inequality,''
\newblock {\em Mathematics of operations research}, vol. 35, no. 2, pp. 438--457, 2010.


\bibitem{attouch2013}
H. Attouch, J. Bolte, and B. F. Svaiter,
\newblock ``Convergence of descent methods for semi-algebraic and tame problems: proximal algorithms, forward--backward splitting, and regularized Gauss--Seidel methods,''
\newblock {\em Math. Program.}, vol. 137, no. 1, pp. 91--129, 2013.



\bibitem{Beck}
A. Beck and M. Teboulle,
`` A fast iterative shrinkage-thresholding algorithm for linear inverse problems,''
{\em  SIAM J. Imaging Sci.},  vol. 2, no. 1,  pp. 183–202, 2009.



\bibitem{Bickel}
P. J. Bickel, Y. Ritov,  and A. B. Tsybakov,
``Simultaneous analysis of Lasso and Dantzig selector'',
{\em Ann. Stat.},  vol. 37, no. 4,   pp. 1705–1732, 2009.


\bibitem{blumensath2009iterative}
T.  Blumensath and M. E.  Davies,
\newblock ``Iterative hard thresholding for compressed sensing,''
\newblock {\em Appl. Comput. Harmon. Anal.},  vol. 27, no.3, pp. :265--274, 2009.



\bibitem{bolte2007}
J. Bolte,  A. Daniilidis, and A. Lewis,
\newblock ``The {\L}ojasiewicz inequality for nonsmooth subanalytic functions with applications to subgradient dynamical systems,''
\newblock {\em SIAM J.  Optim.}, vol. 17, no. 4, pp. 1205--1223, 2007.



\bibitem{bolte20072}
J. Bolte,  A. Daniilidis,  O. Ley, and L. Mazet,
\newblock ``Characterizations of {\L}ojasiewicz inequalities: subgradient flows, talweg, convexity,''
\newblock {\em Trans.  Amer.  Math. Soc.}, vol. 362, no. 6, pp. 3319--3363, 2010.


\bibitem{RIPTcai}
T. Cai, L. Wang, and G. Xu,
“Shifting inequality and recovery of sparse signals,”
{\em IEEE Trans. Signal Process.}, vol. 58, no. 3, pp. 1300–1308, 2010.



\bibitem{RIPTcai9}
T. Cai and A. Zhang,
“Sharp RIP bound for sparse signal and low- rank matrix recovery,”
{\em Appl. Comput. Harmon. Anal.},  vol. 35, pp. 74–93, 2013.




\bibitem{RIPTcaiz}
T. Cai and A. Zhang,
“Sparse representation of a polytope and recovery of sparse signals and low-rank matrices,”
{\em IEEE Trans. Inf. Theory},  vol. 60, no.1, pp. 122--132, 2013.

\bibitem{RIPcandes}
E. J.  Cand\`es,
\newblock ``The restricted isometry property and its implications for compressed sensing,''
\newblock {\em Comptes rendus. Mathematique}, vol. 346, no. 9-10, pp. 589--592, 2008.




\bibitem{candes2006}
E. J.  Cand{\`e}s, J.  Romberg, and T. Tao,
\newblock ``Robust uncertainty principles: Exact signal reconstruction from  highly incomplete frequency information,''
\newblock {\em IEEE Trans. Inf. Theory},  vol. 52, no. 2, pp. 489--509, 2006.



\bibitem{RIPcandes2}
E. J. Cand{\`e}s  and T. Tao,
``The Dantzig selector: statistical estimation when $p$ is much larger than $n$'',
{\em Ann. Stat.},  vol. 35, no. 6, pp. 2313–2351, 2007.


\bibitem{CJT}
E. J. Cand{\`e}s, J. Romberg, and T. Tao,
\newblock ``Stable signal recovery from incomplete and inaccurate measurements,''
\newblock {\em Comm. Pure Appl. Math.},  vol. 59, no. 8, pp. 1207-1223, 2006.




\bibitem{CY}
R. Chartrand and  W. Yin,
''Iteratively reweighted algorithms for compressive sensing,''
in {\em 2008 IEEE International Conference on Acoustics, Speech and Signal Processing,}  pp. 3869–3872, 2008.



\bibitem{chen2001}
S. S.  Chen, D. L.  Donoho, and M. A.  Saunders,
\newblock ``Atomic decomposition by basis pursuit,''
\newblock {\em SIAM review},  vol. 43, no. 1, pp. 129--159, 2001.



\bibitem{Dai}
W. Dai and O. Milenkovic,
``Subspace pursuit for compressive sensing signal reconstruction,''
{\em IEEE Trans. Inf. Theory},  vol. 55, no.5,  pp. 2230–2249, 2009.



\bibitem{donoho}
D.~L.  Donoho,
\newblock ``Compressed sensing,''
\newblock {\em IEEE Trans. Inf. Theory}, vol. 52, no. 4, pp. 1289--1306, 2006.


\bibitem{donoho2}
D. L. Donoho and X. Huo,
``Uncertainty principles and ideal atomic decomposition,''
{\em IEEE Trans. Inf. Theory},  vol. 47, no.7,  pp. 2845–2862, 2001.


\bibitem{du2017}
S. S. Du, C. Jin, J. D. Lee, M. I. Jordan, A. Singh, and B. Poczos,
“Gradient descent can take exponential time to escape saddle points,”
{\em Advances in neural information processing systems},  vol. 30, 2017.


\bibitem{Duarte}
M. F. Duarte,  M. A. Davenport, D. Takhar, J. N. Laska, T. Sun, K. F. Kelly, and  R. G. Baraniuk,
``Single-pixel imaging via compressive sampling,''
{\em  IEEE Signal Process. Mag.}, vol. 25 no. 2, pp. 83–91, 2008.




\bibitem{Ender}
J. H. Ender,
``On compressive sensing applied to radar'',
{\em Signal Process.},  vol. 90, no. 5,  pp. 1402–1414, 2010.


\bibitem{foucart2011hard}
S.  Foucart,
\newblock ``Hard thresholding pursuit: an algorithm for compressive sensing,''
\newblock {\em SIAM Journal on Numerical Analysis}, vol. 49, no. 6, pp. 2543--2563, 2011.

\bibitem{foucar}
S.  Foucart and H.  Rauhut,
\newblock {\em A mathematical introduction to compressive sensing},  vol. 1,
\newblock Birkh{\"a}user Basel, 2013.



\bibitem{Herman}
M. A. Herman and T. Strohmer,
``High-resolution radar via compressed sensing'',
{\em IEEE Trans. Signal Process.},  vol. 57, no. 6,  pp. 2275–2284, 2009.



\bibitem{hoff2017}
P. D. Hoff,
“Lasso, fractional norm and structured sparse estimation using a Hadamard product parametrization,”
{\em Comput. Stat.  Data  Anal.}, vol. 115, no. 3, pp. 186--198, 2017.




\bibitem{lai2018spark}
C. K. Lai, S. Li, and D. Mondo,
“Spark-level sparsity and the $\ell_1$ tail minimization,”
{\em Appl. Comput. Harmon. Anal.}, vol. 45, no. 1, pp. 206--215, 2018.






\bibitem{lee2019}
J. D. Lee, I. Panageas, G. Piliouras,  M. Simchowitz, M. I. Jordan, and B. Recht,
\newblock ``First-order methods almost always avoid strict saddle points,''
\newblock {\em  Math. Program.},   vol. 176,  pp. 311--337, 2019.






\bibitem{Guangxi}
G. Li, S. Li, D. Li, and C. Ma,
\newblock ``The tail-Hadamard product parametrization algorithm for compressed sensing,''
\newblock {\em  Signal Processing},   vol. 205,  p. 108853, 2023.


\bibitem{li2019alternating}
Q. Li, Z. Zhu, and G. Tang,
\newblock ``Alternating minimizations converge to second-order optimal solutions,''
\newblock {\em  International Conference on Machine Learning}, PMLR,  pp. 3935--3943, 2019.








\bibitem{lojasiewicz}
S. Lojasiewicz,
\newblock ``Une propri{\'e}t{\'e} topologique des sous-ensembles analytiques r{\'e}els,''
\newblock {\em  Les {\'e}quations aux d{\'e}riv{\'e}es partielles},   vol. 117, pp. 87--89, 1963.



\bibitem{Lustig}
M. Lustig, D. Donoho, and J. M. Pauly,
``Sparse MRI: the application of compressed sensing for rapid mr imaging'',
{\em Magn. Reson. Med.},  vol. 58,  no.6,   pp. 1182–1195, 2007.

\bibitem{Lustig2}
M. Lustig, D. L. Donoho, J. M. Santos, and J. M. Pauly,
``Compressed sensing MRI'',  {\em IEEE Signal Process. Mag.},  vol. 25, no. 2,  pp. 72–82, 2008.




\bibitem{jin2017}
C. Jin, R. Ge, P. Netrapalli, S. M. Kakade, and M. I. Jordan
\newblock ``How to escape saddle points efficiently,''
\newblock {\em  International conference on machine learning},   PMLR, pp. 1724--1732, 2017.



\bibitem{kurdyka}
K. Kurdyka,
\newblock ``On gradients of functions definable in o-minimal structures,''
\newblock {\em  Annales de l'institut Fourier},   vol. 48, no.3, pp. 769--783, 1998.


\bibitem{Mo}
Q. Mo and S. Li,
“New bounds on the restricted isometry constant $\delta_{2k}$,”
{\em Appl. Comput. Harmon. Anal.}, vol. 31, no. 3, pp. 460–468, 2011.


\bibitem{natarajan}
B. K.  Natarajan,
\newblock ``Sparse approximate solutions to linear systems,''
\newblock {\em SIAM Journal on Computing},  vol. 24, no. 2, pp. 227--234, 1995.


\bibitem{Ouyang}
W.  Ouyang, Y. Liu,  T. K. Pong,  and  H. Wang,
\newblock ``Kurdyka-{\L}ojasiewicz exponent via Hadamard parametrization,''
\newblock {\em arxiv: 2402:00377},  2024.





\bibitem{needell2009cosamp}
D.  Needell and J.~A. Tropp,
\newblock ``CoSaMP: Iterative signal recovery from incomplete and inaccurate samples,''
\newblock {\em Appl. Comput. Harmon. Anal.},  vol. 26, no. 3, pp. 301--321, 2009.



\bibitem{poon2023}
C. Poon and G. Peyr{\'e},
\newblock ``Smooth over-parameterized solvers for non-smooth structured optimization,''
\newblock {\em  Math. Program.},   vol. 201, no.1, pp. 897--952, 2023.



\bibitem{rauhut}
H. Rauhut,
\newblock ``Compressive sensing and structured random matrices,''
\newblock {\em Theoretical foundations and numerical methods for sparse  recovery},  vol. 9, pp. 1--92, 2010.




\bibitem{Recht}
B. Recht, M. Fazel, and P. A. Parrilo,
``Guaranteed minimum-rank solutions of linear matrix equations via nuclear norm minimization,''
{\em SIAM Rev.},  vol. 52, no. 3, pp. 471–501, 2010.


\bibitem{schmidt}
M. Schmidt, G. Fung, and R. Rosales,
``Fast optimization methods for l1 regularization: A comparative study and two new approaches,''
{\em Machine Learning: ECML 2007: 18th European Conference on Machine Learning, Warsaw, Poland},  pp. 286--297, 2007.





\bibitem{Sloun}
R. V.  Sloun, A. Pandharipande, M. Mischi,  and L. Demi,
``Compressed sensing for ultrasound computed tomography'',
{\em IEEE Trans. Biomed. Eng.},  vol. 62, no. 6, pp. 1660–1664, 2015.


\bibitem{Tibshirani}
R. Tibshirani,
``Regression shrinkage and selection via the lasso,''
{\em  J. R. Stat. Soc.}, vol. 58 no. 1, pp. 267–288, 1996.

\bibitem{OMMP}
Z. Xu,
\newblock ``The performance of orthogonal multi-matching pursuit under RIP,''
\newblock  {\em J. Comp. Math.}  vol. 33,  pp. 495-516, 2015.

\bibitem{XCX}
Z. Xu,  X. Chang,  F. Xu, and H. Zhang,
`` $L_{1/2}$ regularization: a thresholding representation theory and a fast solver,''
{\em  IEEE Trans. Neural Netw. Learn. Syst.},  vol. 23, no.7, pp. 1013–1027, 2012.


\bibitem{Yin}
P. Yin,  Y. Lou,  Q. He, and J. Xin,
``Minimization of $\ell_{1-2}$ for compressed sensing,''
{\em SIAM Journal on Scientific Computing},  vol. 37, no.1, pp.  A536-A563, 2015.





\bibitem{zhang2011sparse}
T.  Zhang.
\newblock ``Sparse recovery with orthogonal matching pursuit under RIP,''
\newblock {\em IEEE Trans. Inf. Theory}, vol. 57, no. 9, pp. 6215--6221, 2011.

\bibitem{zheng2022mmv}
B. Zheng, C. Zeng, S. Li, and G. Liao,
\newblock ``The MMV tail null space property and DOA estimations by tail-$\ell_{2,1}$ minimization,''
\newblock {\em  Signal Processing},   vol. 194,  p. 108450, 2022.










%
%
%
%
%
%
%
%
%
%
%
%
%
%
%
%
%
%
%
%
%
%
%
%
%
%
%
%
%
%
%
%
%
%
%
%
%
%
%
%
%
%
%
%
%
%
%
%
%
%
%
%
%
%
%
%
%
%
%
%
%
%
%
%
%
%
%
%
%
%
%
%
%
%
%
%
%
%
%
%
%
%
%
%
%
%
%
%
%
%
%
%
%
%
%
%
%
%
%
%
%
%
%
%
%
%
%
%
%
%
%
%
%
%
%
%
%
%
%
%
%
%
%
%
%
%
%
%
%
%
%
%
%
\end{thebibliography}
\end{document}